%% file: main.tex
\def\confversion{0}
\def\ifconf{\ifnum\confversion=1}
\def\ifnotconf{\ifnum\confversion=0}

\documentclass[11 pt, fleqn, reqno]{article}

\def\showauthornotes{0}
\def\showkeys{0}
\def\showdraftbox{0}

\input{macros}

\input{abbrev}

\usepackage{tikz}
\usetikzlibrary{shapes.symbols}
\usepackage{float}
\usepackage{graphicx}
\usepackage{pei}

\usepackage[top=1in, bottom=1in, left=1.25in, right=1.25in]{geometry}

\newcommand{\pnote}[1] {\textcolor{olive}{ {\textbf{Pei: #1}}}}

\begin{document}
\sloppy

\title{Dimension Independent Disentanglers from Unentanglement and Applications}

\author{Fernando Granha Jeronimo\thanks{{\tt IAS \& Simons Institute}. {\tt granha@ias.edu}. This material is based on work supported by the National Science Foundation under Grant No. CCF-1900460.} \and Pei Wu\thanks{{\tt Weizmann Institute of Science}. {\tt pei.wu@weizmann.ac.il}. This material is based on work supported by the National Science Foundation under Grant No. CCF-1900460. Part of the work is done when the author was at IAS and the Simons Institute.}}

\date{}

\maketitle
\draftbox
\thispagestyle{empty}

\input{abstract}

\newpage

\ifnotconf
\pagenumbering{roman}
\tableofcontents
\clearpage
\fi

\pagenumbering{arabic}
\setcounter{page}{1}

\input{intro}

\input{preliminary}

\input{papo}

\input{de_finetti}

\input{prop_testing}

\input{super_swap}

\input{protocol}

\bibliographystyle{alpha}
\bibliography{macros,references}

\end{document}

%% file: macros.tex
\usepackage{xspace,xcolor,enumerate}
\usepackage{amsmath,amssymb}
\usepackage{amsthm}
\usepackage[toc,page]{appendix}
\usepackage{thmtools}
\usepackage{thm-restate}
\usepackage{color,graphicx}
\usepackage{boxedminipage}
\usepackage{makecell}
\usepackage{tabularx}
\usepackage{enumitem}
\usepackage[linesnumbered,ruled,vlined]{algorithm2e}
\ifnum\showkeys=1
\usepackage[color]{showkeys}
\fi

\definecolor{darkred}{rgb}{0.5,0,0}
\definecolor{darkgreen}{rgb}{0,0.35,0}
\definecolor{darkblue}{rgb}{0,0,0.55}

\usepackage[pdfstartview=FitH,pdfpagemode=UseNone,colorlinks,linkcolor=darkblue,filecolor=darkred,citecolor=darkgreen,urlcolor=darkred,pagebackref]{hyperref}

\usepackage[capitalise,nameinlink]{cleveref}
\usepackage[T1]{fontenc}
\usepackage{mathtools,dsfont,bbm}

\ifnum\showauthornotes=1
\newcommand{\Authornote}[2]{{\sf\small\color{red}{[#1: #2]}}}
\newcommand{\Authorcomment}[2]{{\sf \small\color{gray}{[#1: #2]}}}
\newcommand{\Authorfnote}[2]{\footnote{\color{red}{#1: #2}}}
\else
\newcommand{\Authornote}[2]{}
\newcommand{\Authorcomment}[2]{}
\newcommand{\Authorfnote}[2]{}
\fi

\ifnum\showdraftbox=1
\newcommand{\draftbox}{\begin{center}
  \fbox{%
    \begin{minipage}{2in}%
      \begin{center}%
        \begin{Large}%
          \textsc{Working Draft}%
        \end{Large}\\
        Please do not distribute%
      \end{center}%
    \end{minipage}%
  }%
\end{center}
\vspace{0.2cm}}
\else
\newcommand{\draftbox}{}
\fi

\newtheorem{theorem}{Theorem}[section]

\newtheorem{definition}[theorem]{Definition}

\newtheorem{lemma}[theorem]{Lemma}

\newtheorem{corollary}[theorem]{Corollary}
\newtheorem{claim}[theorem]{Claim}
\newtheorem{fact}[theorem]{Fact}

\theoremstyle{remark}
\newtheorem{algo}[theorem]{Algorithm}

\def\FullBox{\hbox{\vrule width 6pt height 6pt depth 0pt}}

\def\qedsketch{\ifmmode\Box\else{\unskip\nobreak\hfil
\penalty50\hskip1em\null\nobreak\hfil$\Box$
\parfillskip=0pt\finalhyphendemerits=0\endgraf}\fi}

\def\to{\rightarrow}
\def\eps{\varepsilon}
\def\epsilon{\varepsilon}

\def\eps{\epsilon}

\def\cal{\mathcal}

\renewcommand{\bar}{\overline} 

\newcommand{\ie}{i.e.,\xspace}
\newcommand{\eg}{e.g.,\xspace}

\newcommand{\mper}{\,.}
\newcommand{\mcom}{\,,}

\newcommand{\R}{{\mathbb R}}

\newcommand{\C}{{\mathbb C}}
\newcommand{\N}{{\mathbb{N}}}

\newcommand{\cA}{\mathcal{A}}
\newcommand{\cB}{\mathcal{B}}
\newcommand{\cC}{\mathcal{C}}
\newcommand{\cH}{\mathcal{H}}
\newcommand{\cD}{\mathcal{D}}

\newcommand{\cS}{\mathcal{S}}

\usepackage{nicefrac}

\newcommand{\abs}[1]{\ensuremath{\left\lvert #1 \right\rvert}}

\newcommand{\norm}[1]{\ensuremath{\left\lVert #1 \right\rVert}}

\makeatletter

\def\ProbabilityRender#1#2{
  \@ifnextchar\bgroup%
  {\renderwithdist{#1}{#2}}
   {\singlervrender{#1}{#2}}
}
\def\singlervrender#1#2{%
   \ensuremath{\mathchoice
       {{#1}\left[ #2 \right]}
       {{#1}[ #2 ]}
       {{#1}[ #2 ]}
       {{#1}[ #2 ]}
   }
}
\def\renderwithdist#1#2#3{%
   \@ifnextchar\bgroup
   {\superfancyrender{#1}{#2}{#3}}
   {\ensuremath{\mathchoice
      {\underset{#2}{#1}\left[ #3 \right]}
      {{#1}_{#2}[ #3 ]}
      {{#1}_{#2}[ #3 ]}
      {{#1}_{#2}[ #3 ]}
     }
   }
}
\def\superfancyrender#1#2#3#4#5{
   \ensuremath{\mathchoice
      {\underset{#1}{{#1}}\left#4 #3 \right#5}
      {{#1}_{#2}#4 #3 #5}
      {{#1}_{#2}#4 #3 #5}
      {{#1}_{#2}#4 #3 #5}
   }
}
\makeatother

\newcommand{\conv}[1]{\mathrm{conv}\inparen{#1}}

\newfont{\inhead}{eufm10 scaled\magstep1}

\newcommand{\calD}{{\cal D}}

\newcommand{\calH}{{\cal H}}
\newcommand{\calL}{{\cal L}}

\newcommand{\calP}{{\cal P}}

\newcommand{\poly}{{\mathrm{poly}}}

\DeclareMathOperator\supp{supp}

\DeclareMathOperator{\sym}{\operatorname{Sym}}

\renewcommand{\bar}[1]{\ensuremath{\overline{#1}}}

\newcommand{\1}[1]{\mathds{1}{[#1]}}

\newcommand{\inparen}[1]{\left(#1\right)}             

\newcommand{\ket}[1]{\lvert #1\rangle}
\newcommand{\braket}[2]{\left\langle #1 \,\middle\vert\, #2\right\rangle}
\newcommand{\QMA}{\textup{QMA}}
\newcommand{\BQP}{\textup{BQP}}

\newcommand{\NEXP}{\textup{NEXP}}

%% file: abbrev.tex
\DeclareSymbolFont{extraup}{U}{zavm}{m}{n}
\DeclareMathSymbol{\varheart}{\mathalpha}{extraup}{86}
\DeclareMathSymbol{\vardiamond}{\mathalpha}{extraup}{87}

\DeclareMathOperator{\Sym}{Sym}

\DeclareMathOperator{\Tr}{Tr}

\DeclarePairedDelimiter\set{\lbrace}{\rbrace}

\usepackage{tikz}

\newcommand{\SEP}{\textup{SEP}}
\newcommand{\papo}{\textup{PAPO}}

\newcommand{\Symm}[2]{\vee^{#1}(\C^{#2})}

%% file: abstract.tex
\begin{abstract}    
  Quantum entanglement, a distinctive form of quantum correlation, has become
  a key enabling ingredient in diverse applications in quantum computation, complexity, cryptography, etc.
  However, the presence of unwanted adversarial entanglement also poses
  challenges and even prevents the correct behaviour of many protocols and
  applications.

  In this paper, we explore methods to ``break'' the quantum correlations.
  Specifically, we construct a \emph{dimension-independent} $k$-partite
  disentangler (like) channel from bipartite unentangled input.
  In particular, we show: For every $d,\ell\ge k \in \mathbb{N}^+$, there is an efficient channel
  $\Lambda \colon \mathbb{C}^{d\ell} \otimes \mathbb{C}^{d\ell} \to \mathbb{C}^{dk}$
  such that for every bipartite separable density operator $\rho_1\otimes \rho_2$, the output  $\Lambda(\rho_1\otimes\rho_2)$ is close to a $k$-partite separable state. Concretely, for some distribution $\mu$ on states from $\C^d$,
  $$
    \norm{\Lambda(\rho_1 \otimes \rho_2) - \int \lvert \psi \rangle \langle \psi \rvert^{\otimes k} d\mu(\psi)}_1 \le \tilde O \left(\left(\frac{k^{3}}{\ell}\right)^{1/4}\right).
  $$
  Moreover,  $\Lambda(\lvert \psi \rangle \langle \psi \rvert^{\otimes \ell}\otimes \lvert \psi \rangle \langle \psi \rvert^{\otimes \ell}) = \lvert \psi \rangle \langle \psi \rvert^{\otimes k}$.
  Without the bipartite unentanglement assumption, the above bound is conjectured to be impossible and would imply $\QMA(2)=\QMA$.

  Leveraging multipartite unentanglement ensured by our disentanglers, we achieve the following:
  (i) a new proof that $\QMA(2)$ admits arbitrary gap amplification;
  (ii) a variant of the swap test and product test with
  improved soundness, addressing a major limitation of their original versions.
  More importantly, we demonstrate that unentangled quantum proofs
  of almost general real amplitudes capture $\textup{NEXP}$, thereby greatly relaxing the non-negative amplitudes
  assumption in the recent work of $\QMA^+(2)=\NEXP$ [Jeronimo and Wu, STOC 2023].
  Specifically, our findings show that to capture $\NEXP$, it suffices to have unentangled proofs of the form
  $\lvert \psi \rangle = \sqrt{a} \lvert \psi_+ \rangle + \sqrt{1-a} \lvert \psi_{-} \rangle$ where  $\lvert \psi_+ \rangle$
  has non-negative amplitudes, $\lvert \psi_{-} \rangle$ only has negative amplitudes and
  $\lvert a-(1-a) \rvert \ge 1/\textup{poly}(n)$ with $a \in [0,1]$.
  Additionally, we present a protocol achieving an almost largest possible completeness-soundness gap before obtaining
  $\textup{QMA}^{\mathbb{R}}(k)=\textup{NEXP}$, namely, a $1/\textup{poly}(n)$ additive improvement to the gap results
  in this equality.
\end{abstract}

%% file: intro.tex
\section{Introduction}

Quantum entanglement is a fundamental form of quantum correlation that
can be stronger than any classical correlation~\cite{EPR35, B64,CHSH69,JNVWY20}.
It plays a crucial role in a myriad of areas such as quantum computing, quantum information,
quantum complexity, quantum cryptography, condensed matter physics, etc~\cite{HHHH09,NC10,Watrous18}.
Hence, comprehending both the capabilities and constraints of quantum entanglement stands
as a crucial research endeavor.
However, entanglement can also pose challenges in numerous applications, such as quantum
key distribution and quantum proof systems~\cite{Renner08,KMY03,OW16,HHJWY17}.
This raises the natural question of designing quantum channels that convert quantum states
into unentangled states. 
For the purpose of applications, such channel, also called \emph{disentangler}, $\Phi: \cH \to \cK\otimes\cK$ can be defined to satisfy two
conditions: (i) for any $\ket\psi\in\cK$, there is preimage $\ket\phi$, such that $\Phi(\phi)=\psi\otimes\psi$;
and (ii) for any density operator $\phi\in \cH$, $\Phi(\phi)$ is close to \emph{separable}. 

The \emph{quantum de Finetti} type theorems~\cite{christandl2007one,KR05,Renner08} provide examples of
disentanglers.
A quantum de Finetti theorem quantifies the closeness of a \emph{permutation-invariant} $\ell$-partite
quantum state, to $k$-partite separable states when all but $k$ subsystems are traced out.
A standard quantum de Finetti theorem reads
\begin{theorem}[Quantum de Finetti~\cite{KR05}]\label{thm:std_quantum_de_finetti}
  For every $d,\ell\ge k \in \mathbb{N}^+$, the channel $\Lambda \colon (\mathbb{C}^{d})^{\otimes \ell} \to (\mathbb{C}^d)^{\otimes k}$
  defined as $\Lambda(\rho) = \Tr_{\ell-k}(1/\ell! \sum_{\pi \in \Sym_\ell} \pi \rho \pi^{\dag})$ satisfies
  \begin{align*}
    \norm{\Lambda(\rho) - \int \ket{\psi}\bra{\psi}^{\otimes k} d\mu}_1 \le \frac{2 k d^2}{\ell}.
  \end{align*}
\end{theorem}
Note that the error bound scales at least\footnote{If instead of making the state permutation invariant,
we project it onto the symmetric subspace, which is a perfectly valid and efficient operation in the
quantum setting, then the dependence on $d$ in~\cref{thm:std_quantum_de_finetti} improves from $d^2$ to $d$.}
as $d/\ell$, and in this version of the quantum de Finetti theorem, the parameters are known to be essentially
tight.
Consequently, if each subsystem is composed of $n$ qubits, then obtaining
a non-trivial error bound requires at least $\ell \ge d = 2^{n}$
subsystems, making this channel impractical for many applications.
This is conjectured to be essentially the best you can achieve. In particular,
it is conjectured that for any disentangler, the input dimension will be
exponential in the output dimension~\cite{ABDFS08} to achieve that the output is always $\epsilon$ close in trace distance to some separable states for any constant $\epsilon<1$. 

\paragraph{Dimension Independent Disentangler from Unentanglement}
While the original disentangler conjecture remains widely open, in this work,
we show that there is an explicit, efficient (BQP), and \emph{dimension independent}
quantum disentangler for $k$-partite (output) system starting from a bipartite
unentangled system. More precisely, we prove

\begin{restatable}[Disentangler from unentanglement]{theorem}
{TheoChannel}\label{thm:channel}
  Let $d,\ell \ge k \in \mathbb{N}^+$.
  There is an efficient channel $\Lambda \colon (\mathbb{C}^{d})^{\otimes \ell}\otimes (\mathbb{C}^{d})^{\otimes \ell} \to (\mathbb{C}^d)^{\otimes k}$
  such that for any density operators $\rho_1, \rho_2 \in \mathbb{C}^{d\ell}$ there is a distribution $\mu$ on pure states $\ket\psi \in \C^d$ satisfying
  \begin{align*}
    \norm{\Lambda(\rho_1 \otimes \rho_2) - \int \ket{\psi}\bra{\psi}^{\otimes k} d\mu}_1 \le \tilde O\left(\left(\frac{k^3}{\ell}\right)^{1/4}\right).
  \end{align*}
  Furthermore, product states of the form $\rho_1 = \rho_2 = \ket{\psi}\bra{\psi}^{\otimes \ell}$ are
  mapped to $\ket{\psi}\bra{\psi}^{\otimes k}$.
\end{restatable}

In contrast to the de Finetti disentangler, our disentangler from unentanglement features
error parameters that are independent of the input dimension entirely! 
Subsequently, we discuss applications of~\cref{thm:channel} in testing product states and
the gap amplification in quantum proof systems, culminating in a near-optimal gap amplification for the $\QMA^+(k)$ class: Any improvement on this gap amplification would imply $\QMA^\R(k)=\NEXP$.\footnote{
We don't want to distract the readers by the issue about quantum states over real or complex numbers.
In many cases, quantum computation over reals captures that over complex numbers. However, to the best of the authors' knowledge, this is unclear in the context of $\QMA(2)$. We have to use $\QMA^\R(k)$ to denote the proof systems where the proofs are guaranteed to have real amplitudes.}
We also anticipate that our tool will find further applications beyond those discussed in
this paper.

\subsection{Super Product Test}

The \emph{product test} was designed to test if a state $\ket{\phi}$ is close to $k$-partite
product state, \ie $\ket{\phi} \approx \ket{\phi_1} \otimes \cdots \otimes \ket{\phi_k}$, given
two copies of $\ket{\phi}$.
This test involves applying a sequence of swap tests to each of the $k$ subsystems of the two
copies $\ket\phi$.
Clearly, if $\ket\phi$ is indeed a $k$-partite product state, all the swap tests accept with certainty.
On the other hand, if $\ket\phi$ is entangled across the $k$ subsystems, some swap test will reject
with a probability that depends on the amount of entanglement.
It can be argued that the product test is optimal for ensuring perfect completeness, i.e., accepting
product states with certainty~\cite{HM13}. 

Despite its utility and elegance, the product test has two limitations.
Firstly, it only provides a guarantee concerning its input $\ket{\phi}\otimes \ket{\phi}$ which 
are destroyed after the test, yielding a single classical bit as output.
Very often in applications, one also needs some extra certified input states $\ket\phi$ to manipulate in subsequent computations after
the test.
Secondly, and  probably more irritatingly, the product test always accepts with some 
constant probability (say $\ge 1/2$) no matter how far $\ket{\phi}$ is from
being $k$-partite product, i.e., it has poor soundness.
These limitations can be resolved if you have more than 2 copies of $\ket\phi$~\cite{Kada_2008, SY23}.
For instance, given $\ell$ copies of $\ket\phi$, then one can adapt the product test to
sequentially apply projections on to symmetric subspace on the first, second, and subsequent
subsystems of all the copies of $\phi$.
Intuitively, this should give us a stronger test whose analysis was left as an open problem in~\cite{HM13}.
Recently, She and Yuen~\cite{SY23} analyzed this higher order version of product test achieving improved soundness.
We restate this higher order product test as relying on some $\ell$ unentangled equal copies of 
$\ket{\psi}$ to deduce a $k$-partite product structure of the input state.
One can require something even stronger on the input to achieve what we call \emph{super product test}.

\begin{restatable}{lemma}{LemSupProd}\label{lem:super_product}
  The super product test on input $\ket{\psi} \otimes (\ket{\phi_1}\ldots\ket{\phi_k})^{\otimes \ell}$ accepts with probability
  \[
  \frac{\ell}{(\ell+1)} \cdot \abs{\braket{\psi}{\phi_1}\ldots\ket{\phi_k}}^2 + \frac{1}{(\ell+1)}\mper
  \]
\end{restatable}

This super product test focuses on determining whether a target state $\ket\psi$ is a product state or not.
In addition to the target state, there are $\ell$ copies of an already $k$-partite product state that come to help.
This test is very natural and simple, except it seems to ask too much of its inputs: To compare, the high-order
product test requires some copies of a state whereas \cref{lem:super_product} requires some copies of an already
$k$-partite product state of the form $\ket{\phi_1}\otimes \cdots \otimes \ket{\phi_k}$.
We claim the super product test is not really asking for too much because our disentangler channel effectively
``amplifies'' the number of unentangled systems.
In particular, we can rely on just two unentangled proofs to enforce a state close to 
$(\ket{\phi_1}\otimes \cdots \otimes \ket{\phi_k})^{\otimes \ell}$ by~\cref{thm:channel}. 
For simplicity, consider $k$ unentangled pairs of untangled proofs where the $i$\textsuperscript{th} pair applied \cref{thm:channel} yields $\ket{\phi_i}^{\otimes \ell}$. 
Then run the super product test on a target state $\ket{\psi}$ and the $\ell$ copies
of already product states from our disentangler.
Furthermore, note that it is very cheap to instead enforce a state close to
$(\ket{\phi_1}\otimes \cdots \otimes \ket{\phi_k})^{\otimes 2\ell}$, allowing us
to reserve the extra $\ell$ copies of $\ket{\phi_1}\otimes \cdots \otimes \ket{\phi_k}$
as once the super product test passes, they can be used in any other computations as
a very good proxy of $\ket\psi$.
With this combination, we achieve arbitrarily good soundness without requiring 
more than $2k$ unentangled states\footnote{Naturally, the $2k$ unentangled states
need to get larger in dimension to achieve better soundness.} while obtaining a
guarantee about the output, rather than having just a single classical bit of output. 

The above discussion on how we apply our disentangler is general, which we summarize in~\cref{sec:prop_testing}.

\subsection{A Gap Amplification for $\QMA^+(2)$ up to Criticality}

Next, we turn to the unentangled quantum proofs, the so-called $\QMA(2)$ class~\cite{KMY03}
and its variants.
First, we provide some background on this subject.

The complexity of $\QMA(2)$ was shown to be closely related to a variety
of quantum and classical computational problems, \eg determining if a
mixed state is entangled given its classical description, as well as,
various forms of classical polynomial/tensor optimization
(see~\cite{HM13} for a more comprehensive list).
Despite considerable interest and effort (\eg~\cite{DPS04,ABDFS08,BT09,Beigi10,BCY11,GNN12,SW12,P12,BH13,BH15,HNW17}),
we still only know the trivial complexity bounds $\QMA \subseteq \QMA(2) \subseteq \NEXP$.

Even the fact that $\QMA(2)$ admits strong gap amplification is non-trivial and remained open for
about $10$ years before the seminar work of Harrow and Montanaro~\cite{HM13}.
With \cref{thm:channel}, it is easy to give a new proof of this fact, which we discuss
in~\cref{sec:prop_testing}. 

A variant of $\QMA(2)$, denoted $\QMA^+(2)$, with proofs of nonnegative amplitudes
was introduced by Jeronimo and Wu in~\cite{JW23}.
The goal of this variant was to capture many properties of $\QMA(2)$ while having more
structure in order to obtain a greater understanding.
Indeed, they showed that $\QMA^+(2)=\NEXP$ by designing a $\QMA^+(2)$ protocol for a $\NEXP$-complete problem with a constant gap.
On the other end of their result is the observation that $\QMA^+(2)\subseteq\QMA(2)$
provided that the completeness-soundness gap of $\QMA^+(2)$ is a sufficiently large constant.
This makes $\QMA^+(2)$ an intriguing class to study since either (i) showing that $\QMA^+(2)=\QMA(2)$, via possibly a gap amplification approach for $\QMA^+(2)$,
would characterize the complexity of $\QMA(2)$, or (ii) showing $\QMA^+(2) \ne \QMA(2)$ would
give a better upper bound $\QMA(2) \subsetneq \NEXP$.

By virtue of the unentanglement assumption of $\QMA^+(2)$ and the product test~\cite{HM13},
$\QMA^+(2)$ admits some non-trivial gap amplification.
For example, a gap of $1/\poly(n)$ can be amplified to a constant gap in which the completeness
becomes $1-\exp(-\poly(n))$ and the soundness becomes some constant strictly less than $1$.
Recently, Bassirian, Fefferman and Marwaha~\cite{bassirian2023quantum}, building on~\cite{JW23},
curiously showed that $\QMA^+(1)=\NEXP$ also with a constant gap.\footnote{It is not clear that
their gap can be made as large as the one for $\QMA^+(2)=\NEXP$.}
Since in the large constant gap regime of $\QMA^+(1)$, we have $\QMA^+(1)=\QMA \subseteq \textup{PP}$,
their result rules out the strong gap amplification for $\QMA^+(1)$ unless $\NEXP\subseteq\textup{PP}$.
Moreover, it also suggests that strategies aimed at amplifying the gap for $\QMA^+(2)$ must rely on
the unentanglement assumption.
This is precisely where the tools like the product test or our disentangler become essential.

With our disentangler, we make progresses towards understanding of $\QMA^+(2)$ versus $\QMA(2)$.
In particular, our progresses can be summarized as two aspects with two motivating questions.

\begin{center}
{\it  Motivating question 1. How crucial is the nonnegative amplitudes assumption to obtain $\QMA^+(2)=\NEXP$?}
\end{center}
Regarding our first motivating question, we show that the nonnegative amplitudes assumption can be almost
completely removed by considering unentangled quantum proofs of almost general \emph{real} amplitudes.
More precisely, we show that to capture $\NEXP$ it suffices to have unentangled proofs of the form
$\ket{\psi} = \sqrt{a} \ket{\psi_+} + \sqrt{1-a} \ket{\psi_{-}}$ where
$\ket{\psi_+}$ has nonnegative amplitudes, $\ket{\psi_{-}}$ only has
negative amplitudes and $\abs{a-(1-a)} \ge 1/\poly(n)$ with $a \in [0,1]$.
In words, we require the proofs to have slightly more $\ell_2$-probability mass ($1/\poly(n)$ extra mass) either on
nonnegative or negative amplitudes.
We refer to the quantity $\abs{a-(1-a)}$ as the $\ell_2$-sign bias of $\ket{\psi}$.
We call the associated complexity class $\textup{almost-QMA}^\R(k)$ (formally defined in~\cref{def:almost_qma_k}).
Our main complexity result can be stated as follows.

\begin{restatable}{theorem}{TheoAlmostGenAmp}\label{thm:almost_gen_amp}
  $\NEXP = \textup{almost-QMA}^\R(k)$ with unentangled proofs of $\ell_2$-sign bias of\footnote{The letter $n$ represents the
  input size and $b(n)$ is any polynomial time computable function bounded from below by a polynomial, \ie by $1/n^c$ for
  some constant $c > 0$.} $b(n) \ge \poly(1/n)$ and $k=\poly(1/b(n))$.
\end{restatable}

We obtain the above result by investigating the other motivating question: Since the power of $\QMA^+(k)$
ranges from $\NEXP$ to $\QMA(k)$ depending on the gap,

\begin{center}{\it Motivating question 2.
  How much can we amplify the gap of $\QMA^+(k)$?}
\end{center}

We make significant progress addressing this question.
Specifically, we show that a even more relaxed version of $\QMA^+(3)$, featuring
a single proof with nonnegative amplitudes and the other two with general amplitudes,
equals $\NEXP$, with completeness $1-\exp(-\poly(n))$ and soundness $1/2+1/\poly(n)$.
At the first glance, this looks like a ``just so so'' gap amplification.
It is noteworthy that achieving a slightly improved soundness of $1/2-1/\poly(n)$ would
imply $\QMA^\R(3)=\NEXP$.
In particular, if $\QMA^\R(3) \ne \NEXP$, then there is a sharp phase
transition in the complexity around the gap of a half.

\begin{restatable}{theorem}{TheoQMAThree}\label{thm:qma3}
  $\NEXP = \textup{QMA}^+(3)$ with completeness $c=1-\exp(-\poly(n))$
  and soundness $s=1/2+1/\poly(n)$. Furthermore, we can assume a
  particular case of $\textup{QMA}^+(3)$ in which two unentangled
  proofs have arbitrary amplitudes whereas only one unentangled proof has
  nonnegative amplitudes.
\end{restatable}

\begin{figure}[ht]
  \captionsetup{width=.8\linewidth}
  \centering
\begin{tikzpicture}[scale=1]  
    \draw[|->, thick] (0,0) -- (10,0);
    \node[left] at (0,0) {$0$};
    \node[right] at (10,0) {$1$};

    \draw[thick] (5,-0.1) -- (5,0.1) node[above] {$\frac{1}{2}$};

    \draw[|<->|, yshift=-0.3cm] (4.75,0) -- (5.25,0);
    \node[below] at (5, -0.3) {$\frac{1}{\text{poly}(n)}$};

    \node[above] at (2.5, 0.1) {$\NEXP$};
    \node[above] at (7.5, 0.1) {$\QMA^{\R}(3)$};
\end{tikzpicture}
\caption{Gap and the complexity regime of the particular version of $\QMA^+(3)$ from~\cref{thm:qma3}.
         A gap below $1/2-1/\poly(n)$ corresponds to $\NEXP$, whereas a gap above $1/2+1/\poly(n)$
         corresponds to $\QMA^\R(3)$, illustrating a sharp phase transition.}
\end{figure}
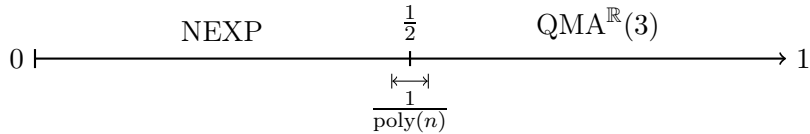

\subsection{Organization}

We introduce notations and review basic concepts and facts in~\cref{sec:prelim}.
In~\cref{sec:papo}, we present an efficient multipartite disentangler (like) channel from bipartite unentanglement.
This construction relies on new de Finetti type properties concerning the interplay between entanglement and symmetry
which we explore in~\cref{sec:de_finetti}.
In~\cref{sec:prop_testing}, we delve into the utility of our disentangler where we elaborate a generic framework in the context of property testing. As one example, we present a new proof that $\QMA(2)$ admits strong gap amplification.
The final two sections are devoted to design new tests and derive the main complexity results in this paper.
In~\cref{sec:super_swap}, we present the super swap and super product test which leverage
unentanglement to achieve much improved soundness than the well-known swap and product tests.
Finally, we provide protocols for $\NEXP$ in~\cref{sec:protocol} leading to the main complexity results of this paper,
\cref{thm:almost_gen_amp} and~\cref{thm:qma3}.

%% file: preliminary.tex
\section{Preliminaries} \label{sec:prelim}

\paragraph{General}
As usual, $\N, \R, \C $ stand for the natural, real, and complex numbers, respectively. 
We adopt the Dirac notation for vectors representing quantum states, e.g., $\ket\psi, \ket\phi$, etc. In this paper, all the vectors of the form $\ket\psi$ are unit vectors. Given any pure state $\ket\psi$, we adopt the convention that its density operator is denoted by the Greek letter without the ``ket'', e.g. $\psi = \ket\psi\bra\psi$. The set of density operators in an arbitrary Hilbert space $\cH$ is denoted $\cD(\cH)$.
A symmetric state $\ket\psi \in (\C^d)^{\otimes n}$ is that invariant under any permutation $\pi\in \sym_n$ where $\sym_n$ is the symmetric group. The action of $\pi$ on $(\C^d)^{\otimes n}$ is  
\[
\pi:|\psi_1, \psi_2, \ldots, \psi_n\rangle
\mapsto |\psi_{\pi(1)}, \psi_{\pi(2)}, \ldots, \psi_{\pi(n)} \rangle.
\]
The \emph{symmetric subspace} is the subspace of $(\C^d)^{\otimes n}$ that is invariant under $\sym_n$, denoted by $\Symm{n}{d}$. Given any set $H\subseteq \cH$ for some Hilbert space $\cH$, $\conv H$ is the convex hull of $H$. 

One other particularly interesting set of states is the \emph{separable} states. We adopt the following notation for the set of density operators regarding separable states,
\begin{align*}
  \SEP(d,r) \coloneqq \conv{\psi_1 \otimes \cdots \otimes \psi_r \mid \ket{\psi_1},\ldots,\ket{\psi_r} \in \C^d}.
\end{align*}
A related notion is that of separable measurement, whose formal definition is given below.
\begin{definition}[Separable measurement]
A measurement $M = (M_0, M_1)$ is separable if in the yes case, the corresponding positive semi-definite matrix $M_1$ can be represented as a conical combination of two operators acting on the first and second parts, i.e.,
for some distribution $\mu$ over the tensor product of positive semi-definite matrices $\alpha$ and $\beta$ on the corresponding space,
\[
M_1 = \int  \alpha \otimes \beta ~\mathrm{d}\mu.
\]
\end{definition}
We record the following well-known fact. An interested reader is referred to~\cite{harrow2013church} for a formal proof.
\begin{fact}[Folklore] \label{fact:swap-test-sep}
The swap test is separable.
\end{fact}

\paragraph{Matrix Analysis} Given any matrix $M\in \C^{n\times n}$, $M^\dagger$ is its  conjugate transpose. Let $\sigma_1 \ge \sigma_2 \ge \ldots \ge \sigma_n$ denote its singular values. Then the trace norm $\|\cdot\|_1$, Frobenius norm $\|\cdot\|_F$ are defined as below
\begin{align*}
    \|M\|_1 = \sum_i \sigma_i, \qquad \|M\|_F = \sqrt{\sum_i \sigma_i^2} .
\end{align*}
The Frobenius norm also equals the square root of sum of squared modulus of each entry, i.e., $\|M\|_F=\sqrt{\sum_{i,j} |M(i,j)|^2}$.

For a positive semi-definite (PSD) matrix $M$, $\|M\|_F = \sqrt{\trace M ^2}$.  
For two PSD matrices, there is one (of many) analogous matrix Cauchy-Schwarz inequality.
\begin{equation}\label{eq:matrix-cauchy-schwarz}
\trace(\sigma\rho)\le \|\sigma\|_F \cdot\|\rho\|_F.
\end{equation}
We adopt the notation $\succeq$ to denote the partial order  that $\sigma \succeq \rho$ if $\sigma-\rho$ is  positive semi-definite.

\paragraph{Distances between Quantum States}
A standard notion of distance for quantum states is that of the \emph{trace distance}. The trace distance between $\psi$ and $\phi$, denoted $\TD(\psi, \phi)$, is 
\begin{align}
\frac{1}{2}\|\psi-\phi\|_1 = \frac{1}{2}\trace\sqrt{( \psi - \phi)^\dagger (\psi-\phi)}.\label{eq:trace-dist}
\end{align}
We also use the notation $\TD(\ket\psi, \ket\phi)$ if we want to emphasize that $\psi$ and $\phi$ are pure states. The following fact provides an alternative definition for trace distance between pure states.
\begin{fact}\label{fact:trace-innerproduct}
  The trace distance between $\ket{\phi}$ and $\ket{\psi}$ is given by $\TD(\ket{\phi},\ket{\psi}) = \sqrt{1-\abs{\braket{\phi}{\psi}}^2}$.
\end{fact}

Two states with small trace distance are indistinguishable to quantum protocols.
\begin{fact}\label{fact:trace_norm_acc}
If a quantum protocol accepts a state $\phi$ with probability at most $p$, then it accepts $\psi$  with
  probability at most $p + \TD(\phi,\psi)$.
\end{fact}

Trace distance enjoys the triangle inequality. For pure states, we can actually strengthen it.

\begin{claim} \label{fact:overlap-triangle-ineq}
Given unit vectors $\ket\alpha, \ket\phi, \ket\beta \in \cH$ for some Hilbert space $\cH$. Suppose
    \begin{align*}
        |\langle \alpha \mid \phi \rangle|^2 = 1-\epsilon, \qquad
        |\langle \beta \mid \phi \rangle|^2 = 1-\delta. 
    \end{align*}
Then for any $\epsilon + \delta \le 1$,\footnotemark
    \begin{equation}      
    |\langle \alpha \mid \beta \rangle|^2 \ge 
    (
        \sqrt{(1-\epsilon)(1-\delta)}-\sqrt{\epsilon\delta}
    )^2.\label{eq:overlap-triangle-ineq-tight}
    \end{equation}
In general, we always have 
    \begin{equation}      
    |\langle \alpha \mid \beta \rangle|^2 \ge 1-\epsilon-\delta-2\sqrt{\epsilon\delta}. \label{eq:overlap-triangle-ineq-crude}
    \end{equation}
\end{claim}
\footnotetext{When $\epsilon+\delta > 1$, then $\ket\alpha$ and $\ket\beta$ in general can be orthogonal.}
\begin{proof}
Without loss of generality assume that
\begin{align}
    &\ket\alpha = \sqrt{1-\epsilon}\ket\phi + \sqrt\epsilon \ket \mu, 
    \nonumber\\
    &\ket\beta = \sqrt{1-\delta}\ket\phi + \sigma\sqrt\eta\ket\mu + \sqrt{\delta-\eta}\ket\rho,
    \nonumber
\end{align}
where $\ket\mu, \ket\rho, \ket\phi$ are orthogonal, $0\le\eta\le \delta$ and $\sigma\in \C$ is a relative phase. Using the basis $\{\ket\phi, \ket\mu, \ket\rho\}$, we can write down explicitly the density matrix of $\alpha$ and $\beta$:
\begin{align*}
&\alpha = \begin{pmatrix}
    1-\epsilon & \sqrt{\epsilon(1-\epsilon)} & 0 \\
    \sqrt{\epsilon(1-\epsilon)} & \epsilon & 0\\
    0 & 0&0
\end{pmatrix},
\\
&\beta = \begin{pmatrix}
    1-\delta   &   \sigma\sqrt{(1-\delta)\eta}  &  
 \sqrt{(1-\delta)(\delta-\eta)} \\
    \sigma^* \sqrt{(1-\delta)\eta} &   \eta    &   \sigma\sqrt{\eta(\delta-\eta)}\\
    \sqrt{(1-\delta)(\delta-\eta)} &   \sigma^*\sqrt{\eta(\delta-\eta)} & \delta-\eta
\end{pmatrix}.
\end{align*}
Now by definition, 
\begin{align}
    \TD(\ket\alpha, \ket\beta) ^2
    &= \left(\frac{1}{2}\trace\sqrt{(\alpha - \beta)^\dagger (\alpha-\beta)}\right)^2
       \nonumber \\
    &= \frac{1}{2} \|\alpha-\beta\|_F^2
        \nonumber\\
    &= \frac{1}{2}( (\epsilon-\delta)^2 + (\epsilon-\eta)^2 + (\delta-\eta)^2 )
        +\eta(\delta-\eta) 
        \nonumber \\
    &\qquad\qquad + |\sqrt{\epsilon(1-\epsilon)}-\sigma\sqrt{(1-\delta)\eta}|^2 +(1-\delta)(\delta-\eta) 
        \nonumber\\
    &\le \frac{1}{2}( (\epsilon-\delta)^2 + (\epsilon-\eta)^2 + (\delta-\eta)^2 )
        +\eta(\delta-\eta) 
        \nonumber \\
    &\qquad\qquad + (\sqrt{\epsilon(1-\epsilon)}+\sqrt{(1-\delta)\eta})^2 +(1-\delta)(\delta-\eta), \label{eq:trace-triangle-bound}
\end{align}
where the second step holds because $\alpha-\beta$ is Hermitian with trace 0 and rank 0 or 2.
We claim that the RHS of (\ref{eq:trace-triangle-bound}), denote by $f$, is non-decreasing for $\eta\in[0,\delta]$. 
By routine calculation, 
\begin{align*}
    \frac{d f}{d\eta}  = -\epsilon + \sqrt{ \frac{\epsilon}{\eta} (1-\epsilon)(1-\delta)} \ge 0 \quad
    \iff \quad(1-\epsilon)(1-\delta) \ge \eta\epsilon \quad
    \\
    \impliedby 
    \quad
    (1-\epsilon)(1-\delta) \ge \delta\epsilon
    \quad
    \iff \quad 1 \ge \epsilon + \delta.
\end{align*}
As we assumed that $1 \ge \epsilon+\delta$, $df/d\eta$ is always non-negative.
Since the RHS of (\ref{eq:trace-triangle-bound}) is non-decreasing for $\eta\in[0,\delta]$, plug $\eta=\delta$ into the RHS of (\ref{eq:trace-triangle-bound}), we obtain
\begin{align*}
    \TD(\ket\alpha,\ket\beta) ^2
    &\le (\epsilon-\delta)^2 + (\sqrt{\epsilon(1-\epsilon)}+\sqrt{(1-\delta)\delta})^2,
\end{align*}
In view of \cref{fact:trace-innerproduct}, (\ref{eq:overlap-triangle-ineq-tight}) is proved. The ``in general'' part is trivially true when $\epsilon+\delta>1$ and otherwise follows from (\ref{eq:overlap-triangle-ineq-tight}).
\end{proof}

Another widely used distance measure between quantum states is that of \emph{fidelity}. For any density operators $\rho, \sigma$ from the same Hilbert space,
\begin{equation*}
    F(\rho, \sigma) = \left(\trace \sqrt{\sqrt\rho \sigma \sqrt\rho)}\right)^2.
\end{equation*}
For our purposes, we only need the fact that when one of the two density operators corresponds to a pure state, then
\begin{equation*}
    F(\rho,\sigma) = \trace(\rho\sigma).
\end{equation*}
The well-known data processing inequality for fidelity states that applying quantum operation never decreases the fidelity. 
\begin{fact}\label{fact:data-processing}
For any quantum channel (CPTP map) $\Phi$,
    \[F(\Phi(\rho), \Phi(\sigma)) \ge F(\rho, \sigma).\]
\end{fact}

\paragraph{Schmidt Decomposition and Partial Trace} 
For $\ket\psi$ describing quantum states over two subsystems $A,B$, e.g., $\ket\psi\in\C^m \otimes \C^n$, there are two sets of orthonormal states $\{\ket{\alpha_1}, \ket{\alpha_2}, \ldots, \ket{\alpha_k}\}\subseteq \C^m$, $\{\ket{\beta_1}, \ket{\beta_2}, \ldots, \ket{\beta_k}\}\subseteq \C^n$ , and positive numbers $\lambda_1\ge \lambda_2\ge \cdots \ge \lambda_k$ for some $k\le \min\{n,m\}$ such that
\begin{equation}
    \ket\psi = \sum_{i=1}^k \sqrt{\lambda_i} \ket{\alpha_i}\ket{\beta_i}, \qquad \text{and } \sum_{i=1}^k \lambda_i = 1.\label{eq:schmidt-decomp}
\end{equation}
The formula (\ref{eq:schmidt-decomp}) is called the \emph{Schmidt decomposition} of $\ket\psi$. The set of $\sqrt{\lambda_i}$ is unique, and is called the \emph{Schmidt coefficient} of $\ket\psi$. 
We  call $\sqrt{\lambda_1}$ the \emph{top Schmidt coefficient} and $\ket{\alpha_1}\ket{\beta_1}$ the \emph{top Schmidt component}. 
Note that the top Schmidt component may not be unique ignoring the global phases, in that case we break  tie arbitrarily. Since Schmidt decomposition follows from singular value decomposition, the (top) Schmidt coefficients can also be formulated as some optimization problem.
\begin{claim}\label{claim:top-schmidt-coeff}
Given any state $\ket\psi\in \cH_1\otimes\cH_2$. Then
\begin{equation*}
    \lambda_1 = \max_{\ket\sigma\in \cH_1, \ket\rho\in\cH_2}|\langle \psi \mid \sigma,\rho \rangle |^2
\end{equation*}
\end{claim}

Often we want to study the density operator of a quantum state $\ket\psi$ over the subsystem $A$, mathematically described by tracing out $B$, denoted  $\trace_B(\psi)$.
We also abbreviate $\psi_A=\trace_B(\psi)$. Note that fidelity never increases under partial trace due to \cref{fact:data-processing}, and similarly, the trace distance never increases under partial trace:
\begin{fact} \label{fact:trace-dist-partial-trace}
For any quantum states $\psi$ and $\phi$ over systems $A$ and $B$,
\[
    \TD(\psi, \phi) \ge \TD(\psi_A, \phi_A).
\]
\end{fact}
We use subscript to emphasize the systems that an operator is describing, e.g., $\psi^{AB}$ simply means that $\psi$ is a state over systems $A$ and $B$.

\paragraph{Quantum Merlin-Arthur Systems}
We now formally define the class $\textup{almost-QMA}^{\R}(k)$, but
first we will need the $\ell_2$-sign bias definition, which, roughly
speaking, quantifies the imbalance in $\ell_2$ mass between the
positive and negative amplitudes parts of a state.

\begin{definition}[$\ell_2$-sign bias]\label{def:ell2_sign_bias}
  Given $\ket{\psi} \in \R^n$, we can uniquely write it as $\ket{\psi}= \sqrt{a} \ket{\psi_+} + \sqrt{1-a} \ket{\psi_{-}}$,
  where $a \in [0,1]$, $\ket{\psi_{+}}$ and $\ket{\psi_{-}}$ are unit vectors with only positive and negative amplitudes, respectively.
  The $\ell_2$-sign bias of $\ket{\psi}$ is defined as $\abs{a-(1-a)}$.
\end{definition}

  Note that a non-negative amplitude state has $\ell_2$-sign bias of $1$ whereas a general
  state has bias at least $0$.
  \textup{Almost}-$\QMA^\R(k)$ will be defined based on $\ell_2$-sign as a natural relaxation of $\QMA^+(k)$ towards the general $\QMA(k)$.

\begin{definition}[$\textup{almost-QMA}^\R(k)$]\label{def:almost_qma_k}
  Let $k \colon \mathbb{N} \to \mathbb{N}$ be a polynomial time computable function.
  A promise problem $\mathcal{L}_{\textup{yes}},\mathcal{L}_{\textup{no}} \subseteq \set{0,1}^*$ is in $\textup{almost-QMA}^\R(k)$ if there exists a $\BQP$
  verifier $V$ such that for every $n \in \mathbb{N}$ and every $x \in \set{0,1}^n$, 
  \begin{itemize}[label={}]
  \item \textbf{\textup{Completeness:}} If $x \in \mathcal{L}_{\textup{yes}}$, then there exist unentangled states $\ket{\psi_1}, \ldots, \ket{\psi_{k(n)}}$, each of $\ell_2$-sign
    bias $1/\poly(n)$ and on at most $\poly(n)$ qubits,
    s.t.\ $\Pr[V(x, \ket{\psi_1} \otimes \cdots \otimes \ket{\psi_{k(n)}})\textup{ accepts}] \ge 9/10$.
\vspace{\parsep}
    \item \textbf{\textup{Soundness:}} If $x \in \mathcal{L}_{\textup{no}}$, then for every unentangled states $\ket{\psi_1}, \ldots, \ket{\psi_{k(n)}}$, each of each of $\ell_2$-sign
    bias $1/\poly(n)$ and on at most $\poly(n)$ qubits,
    we have \ $\Pr[V(x, \ket{\psi_1} \otimes \cdots \otimes \ket{\psi_{k(n)}})\textup{ accepts}] \le 1/10$.
  \end{itemize}
\end{definition}

%% file: papo.tex
\section{The Disentangler from Unentanglement}\label{sec:papo}

In this section, we show how to obtain the dimension independent $k$-partite disentangler (like) channel from bi-partite unentanglement establishing~\cref{thm:channel}.
We will actually work mainly with a more refined procedure which we call quantum probably approximately product output (PAPO) procedure, from which the claimed disentangler can be easily constructed. We define PAPO as follows.

\begin{definition}[PAPO]\label{def:papo}
Let $d, \ell, k \in \N$ and $\epsilon,\delta \in [0,1]$.
A $(d,\ell, k, \epsilon,\delta)$-\textup{PAPO} is a quantum procedure $\Lambda$ satisfying:
    \begin{itemize}[label={}]
        \item \textbf{\textup{Completeness:}} $\forall \ket{\psi} \in \C^d$, $\Lambda(\rho_1 \otimes \rho_2) = \ket{\psi}\bra{\psi}^{\otimes k}$ where $\rho_1=\rho_2=\ket{\psi}\bra{\psi}^{\otimes \ell}$,
        \item \textbf{\textup{Soundness:}} $\forall \rho \in \SEP(d^\ell,2)$, 
   with probability at least $1-\delta$, $\Lambda(\rho)$ either rejects or outputs a state $\epsilon$-close in trace distance to a separable state.
\end{itemize}
\end{definition}

The main result in this section is an efficient PAPO procedure with parameter $\ell$ that is independent of
the dimension $d$. 

\begin{theorem}\label{thm:papo}
  For every $d,k \in \N$ and $\epsilon, \delta \in [0,1]$, there is an efficient $(d,\ell,k,\epsilon,\delta)$-$\papo$
  with $\ell = O(k^3\epsilon^{-2}\delta^{-2}\log\delta^{-1})$.
\end{theorem}

In \cref{proc:papo}, we give a detailed description of our PAPO
procedure. The procedure takes input two unentangled states, each over $\ell$ subsystems. We name the $\ell$ systems $A_1, A_2, \ldots, A_\ell$ for the first state, and $B_1, B_2, \ldots, B_\ell$ for the second state. The PAPO procedure is very simple, which we consider an
advantage for such a fundamental task.  It should be compared with the product test~\cite{HM13}: the PAPO procedure further takes advantage of
symmetric subspace and that projection onto the symmetric subspace is
efficient for quantum algorithms. 

\begin{algorithm}[H]{PAPO}{($\rho^{A_1,A_2,\cdots, A_\ell}\otimes \rho^{B_1, B_2,\cdots, B_\ell} \in \SEP(d^\ell,2)$)}
     \begin{itemize}
      \item Sample $\ell' \in [\ell-k]$ uniformly at random.
      \item For $i=1,\ldots,\ell'$: \begin{enumerate}
           \item Project $\rho^{A_i, \cdots, A_\ell}$ onto the symmetric space.
           \item Project $\rho^{B_i,\cdots, B_\ell}$ onto the symmetric space.
           \item If any of the projections fails:  \emph{Reject}.
           \item If $i\not=\ell'$, SwapTest($\rho^{A_i}, \rho^{B_i}$).
           \item If the SwapTest fails:  \emph{Reject}.
          \end{enumerate}
      \item Output $\rho^{A_{\ell'},\cdots, A_{\ell'+k-1}}$.
     \end{itemize}
     \caption{PAPO}\label{proc:papo}
\end{algorithm}

\subsection{Analysis of PAPO}
The efficiency of the protocol is trivial. Indeed projection onto the symmetric subspace can be implemented efficiently, see for example~\cite{BBDEJM97}, and swap test is a special case of projection onto the symmetric subspace. So in the remainder of the
section, we argue that our procedure satisfies the completeness and
soundness criterion in~\cref{def:papo}. We start with the following
definition of \emph{termination index}.

\begin{definition}[Termination Index]\label{def:term_index}
  We set $i^*$ to be the least element in $[\ell-k]$ such that either $\rho^{A_{i^*},\ldots, A_\ell}$ or  $\rho^{B_{i^*},\ldots, B_\ell}$
  is orthogonal to the symmetric subspace; we set $i^*=\infty$ if no such element exists.
\end{definition}
Here the ``termination'' means absolute termination (rejection) by projection into the symmetric subspace and has nothing to do with a particular execution of \cref{proc:papo}.
Most likely, projecting a general state into the symmetric subspace can success or fail. 
When a state can be successfully projected into the symmetric subspace with nonzero probability, then PAPO  continues to run with nonzero probability. Such case is not counted as absolute termination. \footnote{Note that the swap test has no danger of absolute termination since it is always applied to separable states in \cref{proc:papo} and the swap test has soundness 1/2. Thus in the definition of termination index, we don't worry about the swap test.}

\begin{claim}\label{claim:rho-separable}
  The state $\rho^{A_i,\ldots,A_ \ell, B_i\ldots, B_\ell}$ at the $i$\textsuperscript{th} iteration of the for loop in~\cref{proc:papo} is separable across $\rho^{A_i,\ldots, A_\ell}$ and $ \rho^{B_i,\ldots,B_\ell}$.
\end{claim}
\begin{proof}
Because the SwapTest is separable across $A$ and $B$ part given it accepts by~\cref{fact:swap-test-sep} and projection into the symmetric subspace for $A$ and $B$ part individually is also separable. 
Therefore $\rho^{A_i\ldots A_\ell B_i\ldots B_\ell}$ is separable across $A$ and $B$ part.
\end{proof}

\begin{definition}[Bad Index]
We say that an index $i \in [\ell]$ is $\eta$-\emph{bad}
\begin{enumerate}
    \item If $i \ge i^*$, (see~\cref{def:term_index})
    \item or if $\textup{SwapTest}(\rho^{A_i}, \rho^{B_i})$ accepts with probability at most $1-\eta$.
\end{enumerate}
\end{definition}

\begin{claim}\label{fact:swap-test}
  $\textup{SwapTest}(\rho, \sigma)$ accepts with probability $\frac{1 + \Tr(\rho\sigma)}{2} \le \frac{3}{4} + \frac{\Tr(\sigma^2)}{4}$.
\end{claim}
\begin{proof}
Apply (\ref{eq:matrix-cauchy-schwarz}) for the density operators,
\begin{align*}
    \trace(\rho \sigma) \le \sqrt{\trace(\sigma^2) \cdot \trace(\rho^2)} \le \frac{\trace{\sigma^2}+\trace{\rho^2}}{2},
\end{align*}
where the second step uses the AM-GM inequality. Note that $\trace\rho^2\le 1$, we are done.
\end{proof}

One more technical tool that we are going to need is the following, whose proof we defer to the next section.
\begin{theorem}\label{thm:de-Finnetti-type-1}
Given state $\sigma^{A_1\ldots A_k}\in \conv{\Symm{k}{d}}$. 
Then there is some distribution $\mu$ on pure states $\ket\phi\in\C^d$, such that
\begin{align*}
    \left\|
        \sigma - \int \phi^{\otimes k} d\mu
    \right\|_1 \le O\left(\sqrt{k^3(1-\trace{(\sigma_{A_1})^2}) }\right).
\end{align*}
\end{theorem}

\begin{proof}[Proof of~\cref{thm:papo}]
\textbf{Completeness}: For a desired output of $\ket\psi\bra\psi^{\otimes k}$, we give two unentangled copies of $\ket{\psi}^{\otimes \ell}$ to $\Lambda$ as input.
In this case,~\cref{proc:papo} indeed outputs $\ket\psi\bra\psi^{\otimes k}$ w.p. $1$.
  
\textbf{Soundness}: Let $\rho \in \SEP(d^\ell,2)$ be the input of $\Lambda$. 
Set 
\begin{equation*}
\eta= \epsilon^2 / k^3.
\end{equation*}
Due to \cref{claim:rho-separable}, $\rho^{A_{\ell'}\ldots A_\ell,B_{\ell'}\ldots B_\ell}$ is separable just before the $\ell'$\textsuperscript{th} iteration (assuming successfully reaching this iteration). 
For $\ell'$ that is not a bad index,  after projection onto the symmetric subspace, $\rho^{A_{\ell'} \ldots A_{\ell'+k-1}} \in \conv {\Symm{k}{d}} $.  It follows from \cref{fact:swap-test} that
$\trace_{A_{\ell'}}{(\rho^{A_{\ell'}\ldots A_\ell})^2}\ge 1-4\eta.$ Thus we conclude that if $\ell'$ is not a bad index, then the output (if no rejection)
is $\epsilon$-close in trace distance to a convex combination of product states by~\cref{thm:de-Finnetti-type-1} and our choice of parameter $\eta$. Therefore to prove the theorem, it suffices to bound the probability that~\cref{proc:papo} outputs (not rejects) when $\ell'$ is a bad index.

Next we consider two cases. The first case: If the number of the $\eta$-bad indices among the first $\ell-k$ subsystems are less than $\delta (\ell-k)$,
then with probability at least $1-\delta$, the random index $\ell'$ is not $\eta$-bad. 
Therefore, \cref{def:papo} is satisfied.

The second case: This fraction is larger than $\delta$. 
Now conditioning on the event that $\ell'$ is a bad index, then $\ell'$ is a uniformly random bad index. Therefore, the chance that the set of indices $\{1,2,\ldots, \ell'\}$  contains less than  $\delta/2$
fraction of bad indices is at most $\delta/2$. Thus with probability at least $1-\delta/2,$ we have
seen at least $\delta/2 \cdot \delta (\ell-k) -1$ bad indices in the execution of~\cref{proc:papo} in the first $\ell'$ iterations.
Since for each bad index the probability of not rejecting by the swap test is at most $1-\eta$,
the total probability of not rejecting is at most
\begin{equation}\label{eq:papo-delta}
    (1-\eta)^{\delta^2(\ell-k)-1} = \exp(-\Omega(\eta\delta^2\ell))=
        \exp\left(-\Omega\left(\frac{\ell}{\epsilon^{-2}\delta^{-2} k^3}\right)\right). 
\end{equation}
For 
$\ell = \Omega\left(k^3\epsilon^{-2} \delta^{-2}\log\delta^{-1}\right)$, 
we have $e^{-\eta \delta^2 \ell} \le \delta/2$. In this case, \cref{def:papo} is also satisfied.
\end{proof}  

\subsection{The Disentangler from Unentanglement}
We now construct our disentangler using the PAPO procedure, thereby proving~\cref{thm:channel} (restated below).

\TheoChannel*

\begin{proof}
We set $\epsilon=\delta$, whose exact values will be determined later. Let $\Lambda_0$ be the $(d,\ell,k,\epsilon,\delta)$-$\papo$ procedure guaranteed by~\cref{thm:papo}.
Suppose that we have an input state $\rho \in \SEP(d^\ell,2)$.
The channel $\Lambda$ will be defined as follows. Run the PAPO procedure $\Lambda_0$ on input $\rho$, then
  \begin{enumerate}
    \item If $\Lambda_0(\rho)$ succeeds, $\Lambda$ outputs $\Lambda_0(\rho)$.
    \item Otherwise, $\Lambda$ outputs a fixed product state say $\ket{0}\bra{0}^{\otimes k}$.
  \end{enumerate}
If $\rho = \rho_1 \otimes \rho_2$ with $\rho_1 = \rho_2 = \ket{\psi}\bra{\psi}^{\otimes \ell}$, then $\Lambda$ outputs $\ket{\psi}\bra{\psi}^{\otimes k}$ as desired. If the $\Lambda_0$
rejects, $\Lambda$ outputs a product state. Therefore by the soundness of $\Lambda_0$, firstly, with probability at least $1-\delta$, $\Lambda$ outputs a state $\sigma$ which is $\epsilon$-close to a mixture of product states, i.e., for some distribution $\mu$ on $\calD(\C^d)$,
  \begin{equation*}
      \left\|\sigma-\int _{\ket\psi} \ket\psi\bra\psi ^{\otimes k} d \mu \right\|_1\le \epsilon;
  \end{equation*}
and secondly, with probability $\le\delta$, we output a state $\rho_{\textup{error}}$. 
Overall, we have
  \[
    \Lambda(\rho) = (1-\delta') \sigma + \delta' \rho_{\mathrm{error}}.
  \]
Therefore, 
  \begin{align*}
  \bigg\lVert 
    \Lambda(\rho) - &\int \ket{\psi}\bra{\psi}^{\otimes k}d\mu
  \bigg\rVert _1 
  \\
  &= \norm{(1-\delta') \sigma + \delta' \rho_{\textup{error}} - \int \ket{\psi}\bra{\psi}^{\otimes k}d\mu}_1
  \\
  &\le 
    \norm{\sigma - \int \ket{\psi}\bra{\psi}^{\otimes k}d\mu}_1+
    \norm{-\delta' \sigma + \delta' \rho_{\textup{error}}}_1
  \\
  &\le \epsilon + 2\delta.
  \end{align*}
In view of~\cref{thm:papo}, 
for $\epsilon=\delta$,  
\[
    \left\| 
        \Lambda(\rho) - \int \ket{\psi}\bra{\psi}^{\otimes k}d\mu
    \right\|_1 \le \tilde O\left(
        \left(\frac{k^3}{\ell}\right)^{1/4}\right),
\]
concluding the proof.
\end{proof}

%% file: de_finetti.tex
\section{Quantum Slicing de Finetti Theorem}\label{sec:de_finetti}

In this section, we prove \cref{thm:de-Finnetti-type-1}. In spirit, it is a de Finetti type theorem with the contraint that there is little entanglement across some cut. We refer to such type of theorem as the slicing de Finetti theorem.

\subsection{One-versus-Many Slicing de Finetti}
To start, we study the following most basic scenario that a given permutation-invariant pure quantum state from $\Symm{k}{d}$ has a large top Schmidt coefficient over cut between the first and the remaining subsystems. We obtain a dimension independent quantum de Finetti theorem under slicing constraints from first principles.

\begin{theorem}[One-versus-many Slicing de Finetti]\label{thm:one_vs_many}
Let $\ket{\sigma}^{A_1\ldots A_k} \in \vee^k(\C^d)$. If the largest Schmidt coefficient across
the cut $A_1:A_2\cdots A_k$ is at least $\sqrt{1-\epsilon}$, then
  \begin{align*}
    \max_{\ket{\phi} \in \C^d} \abs{\bra{\sigma}^{A_1\ldots A_k} \ket{\phi}^{\otimes k} }^2  \ge 1 - 8k^3\cdot \epsilon.
  \end{align*}
\end{theorem}

To prove this theorem, we first establish the following duplicate lemma. It says that when a symmetric state $\ket\sigma$ is close to some product state $\ket\phi\ket\rho$, then you can find a new state close to $\ket\sigma$ that with two $\ket\phi$ and harms the closeness only mildly.
\begin{lemma}[Duplicate Lemma]\label{lem:dup}
Let $\ket{\sigma}\in \Symm{k}{d}$. 
Consider some arbitrary decomposition of $\{A_1,A_2,\ldots, A_k\}=\cA\cup \cB \cup \cC$,
such that $|\cA|=|\cB|$. Suppose $ |\bra\sigma ^{\cA\cB\cC} \ket\phi^{\cA} \ket\rho^{\cB\cC}|^2 \ge 1-\epsilon. $
Then, there is a state $\ket\zeta^{\cA\cB\cC}$ such that $\ket\zeta=\ket\phi^\cA\ket\phi^\cB$ $\ket\gamma^{C}$ for some $\ket\gamma^C$, and
    \[
        |\langle\sigma\mid \zeta\rangle|^2 \ge 1-8\epsilon.
    \]
Furthermore, if $\rho_\cC$ is a pure state, then $\gamma = \rho_\cC$.
\end{lemma}

\begin{proof}
We assume that $\epsilon<1/8$, otherwise the statement is trivially true.
Apply Schmidt decomposition to $\ket\rho^{\cB\cC}$ for the $\cB:\cC$ cut,
\[
    \ket\rho^{\cB\cC} = \sum_i \sqrt{\lambda_i} \ket{\beta_i}^\cB\ket{\gamma_i}^\cC.
\]
Let 
\[
    \ket{\rho'}^{\cA\cC}=\sum_i \sqrt{\lambda_i} \ket{\beta_i}^\cA\ket{\gamma_i}^\cC.
\]
Since $\ket\sigma\in\Symm{k}{d}$, we have
\[
    |
        \bra\sigma^{\cA\cB\cC} \ket\phi^\cA\ket\rho^{\cB\cC}|^2
    =
    |
        \bra\sigma^{\cA\cB\cC} \ket\phi^\cB\ket{\rho'}^{\cA\cC}
    |^2
    = 1- \epsilon.
\]
By \cref{fact:overlap-triangle-ineq},
\begin{align}
    (1-2\epsilon)^2
        &\le |\bra\phi^\cA\bra\rho^{\cB\cC} \ket\phi^\cB\ket{\rho'}^{\cA\cC}|^2
        = \left( \sum_i \lambda_i |\langle \phi\mid \beta_i \rangle|^2 \right)^2. \label{eq:slice-ineq}
\end{align}
Abbreviate $\eta_i = |\langle\phi\mid \beta_i \rangle|^2$. Note that
\[
    \sum \eta_i  \le 1, \quad \sum \lambda_i = 1.
\]
Therefore, immediately from (\ref{eq:slice-ineq}),  
\begin{align}
    \lambda_1, \max \eta_i \ge 1-2\epsilon,
\end{align}
which is at least 3/4 since $\epsilon<1/8$. If $\eta_1\not=\max\eta_i$, then 
\[
1-2\epsilon \le \sum_i \lambda_i \eta_i \le \lambda_1 (1-\max \eta_i) + \max \eta_i \cdot (1-\lambda_1) \le 4\epsilon,
\]
which is impossible as $\epsilon < 1/8$. Therefore, $\eta_1 \ge 1-2\epsilon$.

We push it further,
\begin{align}
    1-2\epsilon 
    &\le \sum_i \lambda_i \eta_i \le \lambda_1 \eta_1 + (1-\lambda_1)(1-\eta_1) = 2\lambda_1 \eta_1 -\lambda_1-\eta_1 +1
    \nonumber \\
    &\le 2\lambda_1\eta_1 -2\sqrt{\lambda_1\eta_1} +1
    \nonumber \\
    &=2 \left(\sqrt{\lambda_1 \eta_1 }-\frac{1}{2}\right)^2 + \frac{1}{2},
    \nonumber
\end{align}
where the second step is due to AM-GM inequality. Since $\lambda_1, \eta_1 > 3/4$, and $\epsilon<1/8$, 
\begin{align}
    \lambda_1\eta_1 \ge \left(\frac{1}{2}+\sqrt{\frac{1}{4}-\epsilon}\right)^2 \ge  1-3\epsilon,\nonumber
\end{align}
where the last inequality holds for $\epsilon\in[0,1/8]$. 
Note that
\[
 |\bra\phi^\cA\bra\rho^{\cB\cC} \ket\phi^\cA\ket\phi^\cB\ket{\gamma_1}^\cC|^2\ge\lambda_1\eta_1 \ge 1-3\epsilon.
\]
By \cref{fact:overlap-triangle-ineq} and that $1-3\epsilon > 1/2$, it can be verified that
\begin{align*}
    |\bra\sigma^{\cA\cB\cC} \ket\phi^\cA\ket\phi^\cB\ket{\gamma_1}^\cC|^2 
    &\ge (\sqrt{(1-\epsilon)(1-3\epsilon)}-\sqrt{3}\epsilon)^2
    \\
    &= 1-4\epsilon+6\epsilon^2 -2\sqrt{3}\epsilon\sqrt{(1-\epsilon)(1-3\epsilon)}
    \\
    &\ge 1-8\epsilon.\qedhere
\end{align*}
\end{proof}

Now \cref{thm:one_vs_many} is a simple consequence of \cref{lem:dup}: Duplicate the the first subsystem taken from the top Schmidt component of $\ket\sigma$.
\begin{proof}[Proof of \emph{\cref{thm:one_vs_many}}]
Let $\ket{\sigma_0}=\ket\phi\ket\gamma$ be the top Schmidt component of $\ket\sigma$ for the $A_1: A_2\ldots A_k$ cut. By assumption of the theorem statement,
\[
    |\langle \sigma \mid \sigma_0\rangle|^2 \ge 1-\epsilon.
\]
Let $m =\lfloor \log k \rfloor, m^* = \lceil \log k \rceil$. For $i = 1, 2, \ldots, m$, apply the Duplicate Lemma on $\ket{\sigma_{i-1}}$ with $\cA = \{A_1,A_2,\ldots, A_{2^{i-1}}\}, \cB=\{A_{2^{i-1}+1}, A_{2^{i-1}+2},\ldots, A_{2^i}\}$. Let $\ket{\sigma_i}$ be the $\ket\zeta$ guaranteed by the Duplicate Lemma. 
If $2^m < k$, apply the Duplicate Lemma one more time on $\ket{\sigma_m}$ with $\cA=\{A_1, A_2,\ldots, A_{k - 2^m}\}, \cB=\{A_{2^m+1}, A_{2^m+2}, \ldots, A_k\}$, and let $\ket{\sigma_{m^*}}$ be the state guaranteed by the Duplicate Lemma. 
Then, a straightforward induction shows
\begin{enumerate}
    \item  $|\langle \sigma \mid \sigma_{m^*}\rangle|^2 \ge 1 - 8^{m^*}\epsilon\ge 1-8k^3\epsilon$,
    \item  $\ket{\sigma_{m^*}}=\ket\phi^{\otimes k}$.
\end{enumerate}
That finishes the proof.
\end{proof}

We make a remark about \cref{thm:one_vs_many}. Note that some polynomial dependence on $k$ is unavoidable in this analysis for our procedure. Consider the following state:
\[
    \frac{1}{\sqrt {k+1}}\ket{\vec0}+\frac{1}{\sqrt {k+1}}\sum_{i=1}^{k}\ket{\vec{e_i}}. 
\]
To obtain a tight version of the above theorem with linear dependency on $k$ is an interesting problem.

\subsection{Many-versus-Many Slicing de Finetti}
In \cref{thm:one_vs_many}, we considered top Schmidt coefficient being large on a $1$ vs $k-1$ cut for pure state. By looking at the example we mentioned in the end of the previous subsection, it is natural to think that if the top Schmidt coefficient is large among a balanced cut, then we can obtain better trace distance. That is indeed the case. In fact, that top Schmidt coefficient is large for a balanced cut always implies the top Schmidt coefficient is large for a less balanced cut for a symmetric state. In this subsection, our goal is to formalize this intuition. 

\begin{theorem}[Many-versus-many Slicing de Finetti]\label{thm:many-many}
Let $\ket{\sigma}^{A_1\ldots A_k} \in \vee^k(\C^d)$. Suppose for some  $1\le\ell \le k/2$, the top Schmidt coefficient of $\ket\sigma$ over the $A_1\ldots A_\ell : A_{\ell+1}\ldots A_k$ cut is $\sqrt{1-\epsilon}$. Then there is  $\ket\phi\in\C_d$, such that
\[
    |\langle \sigma, \phi^{\otimes k}\rangle|^2 \ge 1 - O( (k/\ell)^3 \epsilon).
\]
\end{theorem}

We start by collecting a couple of useful facts. The first one says that if a symmetric state from $(\C^d)^{\otimes k}$ is close to a product state, then it is also close to a symmetric product state, i.e., $\ket{\phi^{\otimes k}}$ for some $\ket\phi \in \C ^d$.
\begin{lemma} \label{lem:prod-symmetric-state}
Given a symmetric state $\ket\sigma\in\Symm{k}{d}$ and a $k$-partite product state $\ket\psi\in (\C^d)^{\otimes k}$. Suppose $|\langle \psi \mid \sigma \rangle|^2 \ge 1-\epsilon$. Then there is $\ket\phi\in\C^d$ that satisfies
\[
    |\langle\sigma\mid \phi^{\otimes k}\rangle|^2 \ge 
    1-9\epsilon.
\]
\end{lemma}
\begin{proof}
We take advantage of $\ket\sigma$ being symmetric in a way similar to that of \cref{lem:dup}. As $\ket\sigma\in\Symm{k}{d}$, we have for any permutation $\pi\in \Sym_k$, $|\langle \sigma \mid \pi\psi\rangle|^2\ge1-\epsilon$. By \cref{fact:overlap-triangle-ineq},
\begin{align*}
    |\langle \psi \mid \pi\psi \rangle|^2 \ge 1-4\epsilon.
\end{align*}
Say $\ket\psi=\ket{\psi_1}\otimes\ket{\psi_2}\otimes \cdots \otimes\ket{\psi_k}$, then,
\begin{align}
    (1-4\epsilon)^{k!} 
    \le \prod_{\pi\in \Sym_k} |\langle \psi \mid \pi\psi \rangle|^2    
    &=\left(\prod_{i\in[k]}\prod_{j \in [k]} |\langle \psi_i\mid \psi_j\rangle|^2\right)^{(k-1)!}
        \nonumber\\
    &\qquad\le \left(\Exp_{i\in[k]} \prod_{j\in[k]}
    |\langle \psi_i\mid \psi_j\rangle|^2\right)^{k!},
        \label{eq:prod-sym-ineq}
\end{align}
where the last step uses the AM-GM inequality. It follows from (\ref{eq:prod-sym-ineq}), there must exist $i\in[k]$ such that
\begin{align}
     1-4\epsilon\le \prod_{j\in [k]} |\langle \psi_i \mid \psi_j \rangle|^2 
         \quad \iff \quad
     1-4\epsilon \le |\langle  \psi_i ^{\otimes k} \mid \psi \rangle|^2.
     \nonumber
\end{align}
Apply \cref{fact:overlap-triangle-ineq} one more time, we obtain our lemma.
\end{proof}

The second fact  due to  Harrow and Montanaro~\cite[Appendix~B~Lemma~2]{HM13} and Soleimanifar and Wright~\cite{SW22matrix-product} establishes some criteria when a pure state is close to a product state.
\begin{lemma}
Given any quantum state $\ket\psi\in \cH_1 \otimes \cH_2 \otimes \cdots \otimes \cH_k$ for some arbitrary Hilbert space $\cH_1, \ldots, \cH_k$. Suppose 
\[
    \Exp_{\cS\subseteq [k]} [\trace{\psi_\cS^2}] \ge 1-\epsilon.
\]
Then for some product state $\ket\phi=\ket{\phi_1}\otimes\ket{\phi_2}\otimes\cdots\otimes\ket{\phi_k}$,
\[
    |\langle\psi \mid \phi\rangle|^2 \ge 1-3\epsilon.
\]
\end{lemma}
Combining the above two lemmas, we obtain
\begin{corollary}\label{cor:sym-prod}
Given any  state $\ket\sigma \in \Symm{k}{d}$. Suppose 
\[
    \Exp_{\cS\subseteq [k]} [\trace{\sigma_\cS^2}] \ge 1-\epsilon.
\]
Then for some  state $\ket\phi \in \C^d$,
\[
    |\langle\sigma \mid \phi^{\otimes k}\rangle|^2 \ge 1-27\epsilon.
\]
\end{corollary}

From the above discussion, to prove \cref{thm:many-many}, it suffices to bound $\Tr\psi_\cS^2$ for any subset $\cS$. The following ``cut lemma'' establishes such bounds. 
\begin{lemma}[Cut Lemma]\label{lem:cut}
Let $\ket{\sigma}\in \Symm{k}{d}$.
Suppose for some  $1\le\ell \le k/2$, the top Schmidt coefficient of $\ket\sigma$ over the $A_1\ldots A_\ell : A_{\ell+1}\ldots A_k$ cut is $\sqrt{1-\epsilon}$. Let $\cS\subseteq [k]$ be some arbitrary  subset.
Then, 
\[
    \trace \sigma_\cS ^2 \ge \begin{cases}
        1, & \quad |\cS|  = 0;
        \\
        1-6\epsilon, & \quad \min\{|\cS|, k-|\cS|\} \in \{1,2,\ldots, \ell-1\};
        \\
        1- O((|\cS|/\ell)^3\epsilon), &\quad \min\{|\cS| , k-|\cS|\}\in \{\ell,\ldots, k/2\}.
    \end{cases}
\]
\end{lemma}
\begin{proof}
For $\cS =\varnothing$, the statement is trivial as $\sigma$ is pure. Since $\ket\sigma\in\Symm{k}{d}$, without loss of generality, assume that 
$\cS=\{1, 2, \ldots, m\}$ for some $1 \le m \le k/2$. This is because 
$\trace \sigma_\cS^2 = \trace \sigma^2_{\bar\cS}$ when $\sigma$ is a pure state. Let $\ket\phi ^{A_1\ldots A_\ell} \ket\zeta ^{A_{\ell+1}\ldots A_k}$ be the 
top Schmidt component associated with the coefficient $\sqrt{1-\epsilon}$.

\uline{Case 1}: $m < \ell$. Let $\cA = \{1, 2, \ldots, m\}, \cB = \{m+1, \ldots, \ell\}, \cC = \{\ell+1, \ldots, k-\ell+m \}, \cD = \{k-\ell+m+1, \ldots, k\}.$ Write down the Schmidt decomposition of $\ket\phi$ over the $\cA$ and $\cB$ cut, $\ket\zeta$ over the $\cC$ and $\cD$ cut,
\begin{equation}
    \ket\phi = \sum_i \sqrt{\lambda_i} \ket{\alpha_i}\ket{\beta_i}, 
    \qquad
    \ket\zeta = \sum_i \sqrt{\eta_i} \ket{\gamma_i} \ket{\delta_i}.
    \nonumber
\end{equation}
Since $\cB$ and $\cD$ has the same size, and that $\ket\sigma\in\Symm{k}{d}$, we have for the state $\ket\phi\ket\zeta$, if we switch the subsystem of $\cB$ and $\cD$, then the overlap with $\ket\sigma$ is still $1-\epsilon$. Therefore, by \cref{fact:overlap-triangle-ineq}, we have
\begin{align}
    (1-2\epsilon)^2 
    &\le \left|\left\langle
        \sum_{i,j} \sqrt{\lambda_i\eta_j}\bra{\alpha_i}^\cA\bra{\beta_i}^\cB
         \bra{\gamma_j}^\cC\bra{\delta_j}^\cD,
        \sum_{i,j} \sqrt{\lambda_i\eta_j}\ket{\alpha_i}^\cA\ket{\beta_i}^\cD
        \ket{\gamma_j}^\cC\ket{\delta_j}^\cB
    \right\rangle\right|^2
    \nonumber \\
    & = \left(\sum_{i,j} \lambda_i\eta_j |\langle \beta_i \mid \delta_j \rangle|^2\right)^2
    \le  \lambda_1^2 \left(  \sum_{i,j}\eta_j |\langle \beta_i \mid \delta_j \rangle|^2 \right) ^2
    \nonumber\\
    &\le \lambda_1^2 \left(  \sum_{j}\eta_j \sum_i|\langle \beta_i \mid \delta_j \rangle|^2 \right)^2
    \le \lambda_1 ^2.
    \nonumber
\end{align}
Immediately,
\begin{align}
    \trace{\sigma_\cA ^2 } 
    &\ge (1-\epsilon)^2\trace [(\trace_{\cB\cC\cD}  (\phi\otimes\zeta) )^2]
    \nonumber \\
    &= (1-\epsilon)^2\trace \left[\left(\trace_{\cB\cC\cD} \left(
    \sum_i \lambda_i \alpha_i \otimes \beta_i\otimes\zeta
    \right) \right) ^2\right]
    \nonumber \\
    &\ge (1-\epsilon)^2 \lambda_1 ^2 \ge (1-\epsilon)^2 (1-2\epsilon)^2
    \nonumber\\
    &\ge 1-6\epsilon.\label{eq:many-many-case1-corr}
\end{align} 
The first step is true because $\sigma \succeq (1-\epsilon)\phi\otimes \zeta$, therefore  $\trace_{\cB\cC\cD}\sigma \succeq (1-\epsilon)\trace_{\cB\cC\cD}(\phi\otimes \zeta)$ as partial trace is completely positive. It then follows that $\trace\sigma_\cA ^2 \succeq (1-\epsilon)^2\trace(\trace_{\cB\cC\cD}(\phi\otimes\zeta)) ^2$.

\uline{Case 2}: $\ell < m \le k/2$. We are much like the situation of \cref{thm:one_vs_many}. Let $t = \lceil \log (m / \ell) \rceil$. For $i=1$ to $t$, we apply the Duplicate Lemma and obtain a state $\ket{\sigma_i}$, such that for $i=1,2,\ldots, t$
\begin{align}
    &\trace_{\{\ell\cdot 2^i+1,\ldots, k\}}\sigma_i = \phi ^{\otimes 2^i},
    \nonumber \\
    &\trace_{\{\ell\cdot 2^i+1,\ldots, k\}}\sigma_i^2 = 1,
    \label{eq:many-many-pure} \\
    &|\langle \sigma \mid \sigma_i \rangle|^2 \ge 1-8^i \epsilon.
    \label{eq:many-many-case2-corr}
\end{align}
By our choice of parameter, $2^{t-1}\ell < m \le 2^{t}\ell$. If $m=2^t \ell$, then $(\sigma_t)_\cS = \trace_{\ell\cdot2^t+1,\ldots,k}\sigma_t$ is pure by (\ref{eq:many-many-pure}). Then  
\begin{align*}
    \sqrt {\trace\sigma_\cS^2} =  \sqrt{\trace \sigma_\cS ^2 \cdot \trace (\sigma_t)_\cS^2} &\ge \trace(\sigma_\cS\cdot (\sigma_t)_\cS) = F(\sigma_\cS, (\sigma_t)_\cS) \\
    &\ge F(\sigma, \sigma_t)  = |\langle \sigma \mid \sigma_t\rangle|^2 \ge (1-8^{\log(m/\ell)})\epsilon,
\end{align*}
where the first step and third step are true because $(\sigma_t)_\cS$ is pure; the second step uses (\ref{eq:matrix-cauchy-schwarz}); the fourth step is by \cref{fact:data-processing}, the data processing inequality for fidelity; then fifth step is again by purity of the states; and the final step uses (\ref{eq:many-many-case2-corr}). It follows that
\[
    \trace \sigma_\cS^2 \ge 1 - O((m/\ell)^3\epsilon).
\]
If $m<2^t\ell$, then we can apply Case 1. 
Let $\cA= \{1,2,\ldots,2^t\ell\}, \cB = \{2^t\ell+1, \ldots, k\}$. Then in view of (\ref{eq:many-many-case2-corr}), the top Schmidt coefficient of $\ket\sigma$ among the $\cA : \cB$ cut is at least $\sqrt{1-8^t\epsilon}$ by \cref{claim:top-schmidt-coeff}. Thus by (\ref{eq:many-many-case1-corr}),
\[
    \trace \sigma_\cS^2 \ge 1-6\cdot 8^t\epsilon\ge 1-O((m/\ell)^3\epsilon).
\qedhere
\]
\end{proof}

Now \cref{thm:many-many} follows from \cref{cor:sym-prod} and \cref{lem:cut}.

\subsection{Proof of \cref{thm:de-Finnetti-type-1}}
Now we record a version of the slicing de Finetti theorem for the mixture of symmetric states. A natural generalization of the top Schmidt coefficient among some $A:B$ cut for a state $\sigma$ being large is that $\trace{\sigma_A^2} $ being large.
In particular,
\begin{lemma}\label{lem:trace-reduced-den-squared}
Let $\sigma\in\C^n\otimes\C^m$ be some density operator, and $A,B$ are the systems with respect to the space  $\C^n$ and $\C^m$, respectively. Suppose
\begin{equation*}
    \trace \sigma_A^2 \ge 1-\epsilon.
\end{equation*}
Let $\mu$ be some distribution on pure states induced by $\sigma$, then
\begin{align*}
    \Exp_{\rho\sim\mu} \lambda_1(\rho) \ge 1-\epsilon.
\end{align*}
\end{lemma}
\begin{proof}
Let $m=|\supp\mu|$ be a finite number, this is without loss of generality.  Let $\rho_1, \rho_2,\ldots, \rho_m$ be the pure states in $\supp\mu$. Further, write the Schmidt decomposition for each $\rho_i$
\begin{align*}
    \ket{\rho_i} = \sum_{j} \sqrt{\lambda_{ij}}\ket{\phi_{ij}}^{A}\ket{\sigma_{ij}}^{B},\qquad \lambda_{i1} \ge \lambda_{i2}\ge\cdots.
\end{align*}
Then
\begin{align*}
    \sigma_A = \sum_i \mu(\rho_i) \sum_j \lambda_{ij} \ket{\phi_{ij}}\bra{\phi_{ij}}.
\end{align*}
Thus,
\begin{align}
    \trace{\sigma_{A}^2} 
    &= \sum_i \mu(\rho_i)^2 \sum_j \lambda_{ij}^2 +    \sum_{i\not=i'}\mu(\rho_i)\mu(\rho_{i'})\sum_{j,j'}\lambda_{ij}\lambda_{i'j'}|\langle \phi_{ij} \mid \phi_{i'j'}\rangle|^2
    \nonumber\\
    &\le \sum_i \mu(\rho_i)^2 \sum_j \lambda_{ij}^2 + \sum_{i\not=i'}\mu(\rho_i)\mu(\rho_{i'})\lambda_{i1}\sum_{j,j'}\lambda_{i'j'}|\langle \phi_{ij} \mid \phi_{i'j'}\rangle|^2
    \nonumber\\
    &\le \sum_i \mu(\rho_i)^2 \sum_j \lambda_{ij}^2 + \sum_{i\not=i'}\mu(\rho_i)\mu(\rho_{i'})\lambda_{i1}\sum_{j'}\lambda_{i'j'}
    \nonumber\\
    &\le \sum_i \mu(\rho_i)^2  \sum_j \lambda_{ij}^2 + \sum_{i\not=i'}\mu(\rho_i)\mu(\rho_{i'}) \lambda_{i1}
    \nonumber\\
    &=\sum_i \mu(\rho_i)^2\sum_j \lambda_{ij}^2 + \sum_i \mu(\rho_i)(1-\mu(\rho_i))\lambda_{i1}
    \nonumber\\
    &\le \sum_i \mu(\rho_i)^2 \lambda_{i1} + \sum_i \mu(\rho_i)(1-\mu(\rho_i))\lambda_{i1}
    \nonumber\\
    &= \sum_i \mu(\rho_i)\lambda_{i1},
    \nonumber
\end{align}
where the third step holds because for fixed $i,i',j'$,  $\sum_{j} |\langle\phi_{ij}\mid \phi_{i'j'}\rangle|^2\le 1.$
\end{proof}

\begin{theorem}\label{thm:conv-symm-de-Finetti}
Given density operator $\sigma^{A_1\ldots A_k}$ that describes states from $\conv{\Symm{k}{d}}$. 
For any $1\le \ell\le k/2$ and $\cA=\{A_1,A_2,\cdots, A_\ell\}$, there is some distribution $\mu$ on $\ket\phi\in\C^d$,
\begin{align}
    \left\|
        \sigma - \int \ket\phi\bra\phi^{\otimes k} d\mu
    \right\|_1 \le O\left(\sqrt{(k/\ell)^3(1-\trace{\sigma_{\cA}^2}) }\right).
\end{align}

\end{theorem}
\begin{proof}
Let $\mu$ be the distribution on pure symmetric states induced by $\sigma$. Let $\trace\sigma_{\cA}^2=1-\epsilon$. The theorem follows immediately by combining \cref{fact:trace-innerproduct}, \cref{lem:trace-reduced-den-squared}, \cref{thm:many-many}, and triangle inequality.
\end{proof}

%% file: prop_testing.tex
\section{A Framework: Multiplexing Unentangled States for Property Testing}\label{sec:prop_testing}

In this section, we present a general template illustrating the utility of our disentangler~\cref{thm:channel}. 
We will then use this template multiple of times. Initially, we provide two examples as warm-ups for what is to come. Subsequently, in later sections, we apply this template with carefully designed testers to obtain new complexity results.

Our disentangler leverages a bipartite unentanglement assumption between two states of the form
$\rho_1 \otimes \rho_2$ into an (approximate) multipartite unentanglement assumption of the form $\int \ket\psi\bra\psi^{\otimes k}d\mu$.
Having sufficiently many unentangled copies of a state $\psi$ is particularly important in the context of quantum property
testing as some properties require this assumption for testability. Indeed, many of other information processing
tasks like quantum state tomography often assumes the input is of this form $\ket \psi\bra\psi^{\otimes k}$. Moreover, multiple copies
allow the tester to be executed multiple times amplifying its probability of distinguishing the closeness to
the desired property. Finally, a property tester may end up destroying the copies $\psi^{\otimes k}$ when it measures this state,
so it is desirable to have additional copies that can be used in further information processing tasks once the closeness to
the desired property is certified. In~\cref{fig:dis_and_prop_test}, we provide an illustration of a property
tester being used in conjunction with our disentangler in order to obtain the aforementioned benefits.

\usetikzlibrary{positioning, decorations.pathreplacing, arrows.meta, calc, fit}

\begin{figure}[h]
  \captionsetup{width=.8\linewidth}
  \centering
\begin{tikzpicture}[>=Stealth, scale=0.9]

    \tikzstyle{input} = [draw, minimum width=3.5cm, minimum height=0.6cm, align=center]
    \tikzstyle{output} = [draw, minimum width=0.6cm, minimum height=0.6cm, align=center, inner sep=2pt]
    \tikzstyle{operation} = [draw, minimum width=4.5cm, minimum height=3cm, align=center, rounded corners=2mm]
    \tikzstyle{arrow} = [->, thick]
    \tikzstyle{tester} = [draw, minimum width=2cm, minimum height=1cm, align=center, rounded corners=2mm]

    \node[input] (rho1) {$\rho_1$};
    \node[input, below=of rho1] (rho2) {$\rho_2$};

    \node[operation] at (5,-0.5) (disentangler) {Disentangler};

    \node[output, right=1cm of disentangler] (psi1) {$\psi$};
    \node[right=0.15cm of psi1] (dots1) {$\cdots$};
    \node[output, right=0.15cm of dots1] (psik) {$\psi$};
    
    \node[output, right=0.5cm of psik] (psi1_2) {$\psi$};
    \node[right=0.15cm of psi1_2] (dots2) {$\cdots$};
    \node[output, right=0.15cm of dots2] (psik_2) {$\psi$};
    
    \node[tester, above=0.75cm of psi1, xshift=1cm] (tester) {Property Tester};
    \draw[arrow] (tester) -- (psi1);
    \draw[arrow] (tester) -- (psik);

    \draw[arrow] (rho1) -- (disentangler);
    \draw[arrow] (rho2) -- (disentangler);

    \draw[arrow] (disentangler) -- (psi1);    
\end{tikzpicture}
  \caption{Schematic picture of our disentangler being used to (approximately) ensure multiple unentangled copies of a state as output.
            Part of these copies are used to test a given desired property. If the test passes, the remaining ``certified'' copies can be
            used in further information processing tasks.}\label{fig:dis_and_prop_test}
\end{figure}
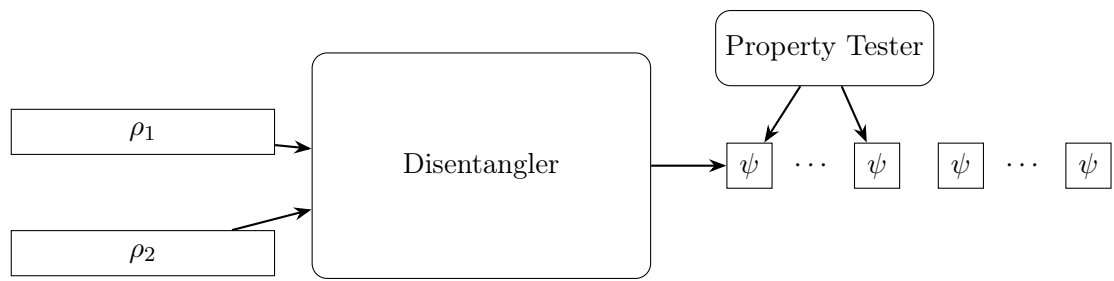

\paragraph{Product Tester and Preparing Multipartite Separable States}
To make this illustration more concrete, first we consider a  scenario where the tester is the product test~\cite{HM13}. More precisely, the product test requires two unentangled copies of $\ket{\psi} \in \C^d$ and checks whether
$\ket\psi$ is close to a  product state of the form $\ket{\phi_1}\otimes \cdots \otimes \ket{\phi_s} \in \C^{d_1} \otimes \cdots \otimes \C^{d_s}$,
where $d=d_1\cdots d_s$. 
For context, 
recall that (an abridged version of) their main result provides the following guarantees for this tester.

\begin{theorem}[Product Test~\cite{HM13}]\label{theo:prod_test}
  Given $\ket{\psi} \in \C^{d_1} \otimes \cdots \otimes \C^{d_s}$, let
  \[
  1-\epsilon = \max \left\{ \abs{\braket{\psi}{\phi_1,\ldots,\phi_s}}^2 \colon \ket{\phi_i} \in \C^{d_i}, 1 \le i \le s  \right\}\mper
  \]
  Let $P_{\text{test}}(\ket{\psi}\bra{\psi})$ be the probability that the product test passes when applied to $\ket{\psi}$.
  Then, we have $P_{\text{test}}(\ket{\psi}\bra{\psi}) = 1 - \Theta(\epsilon)$.
\end{theorem}

Combining our disentangler from~\cref{thm:channel} and the product test from~\cref{theo:prod_test}, we obtain the
following corollary giving all the desired qualities alluded above in a more quantitative way.

\begin{corollary} \label{cor:separable-states}
  Let $\calH = \C^{d_1} \otimes \cdots \otimes \C^{d_s}$. For every $k,k',\ell \in \N$ such that $\ell\ge k+2k'$, there is a channel $\Gamma \colon \calD( \calH^{\otimes \ell} \otimes \calH^{\otimes \ell}) \to
  \calD(\calH^{\otimes k} \otimes \C^2)$ such that for every $\rho_1,\rho_2 \in \calD( \calH^{\otimes \ell})$, there exists $\sigma \in \calD(\calH^{\otimes k} \otimes \C^2)$
  defined as
  \begin{align*}
    \sigma = \int \ket{\psi}\bra{\psi}^{\otimes k} \otimes \left( P_{\text{test}}(\ket{\psi}\bra{\psi})^{k'}  \ket{1}\bra{1} +  (1-P_{\text{test}}(\ket{\psi}\bra{\psi})^{k'}) \ket{0}\bra{0} \right) d\mu\mcom
  \end{align*}
  such that
  \begin{align*}
    \norm{\Gamma(\rho_1 \otimes \rho_2) -  \sigma}_1 \le \tilde O\left(\left(\frac{(k+2k')^3}{\ell}\right)^{1/4}\right)\mper
  \end{align*}
  Furthermore, $\Gamma(\rho_1 \otimes \rho_2) = (\ket{\psi}\bra{\psi})^{\otimes k} \otimes \ket{1}\bra{1}$ provided
  $\rho_1 = \rho_2 = (\ket{\psi}\bra{\psi})^{\otimes \ell}$, where
  $\ket{\psi} = \ket{\phi_1}\otimes \ldots \otimes \ket{\phi_s}$ for some $\ket{\phi_i} \in \C^{d_i}$ for $1 \le i \le s$.
\end{corollary}

\begin{proof}
  Define another channel $\Gamma' \colon \calD(\calH^{\otimes( k+2k')}) \to \calD(\calH^{\otimes k} \otimes \C^2)$ that takes as input
  the output of the disentangler $\Lambda$ which is comprised of $k+2k'$ registers of the space $\calH$. 
  We define the channel $\Gamma'$ to act as identity on the first $k$ registers. On the last $2k'$ registers it performs the product test
  on each pair of registers, outputting a single qubit $\ket{1}\bra{1}$ if all tests pass, otherwise outputting $\ket{0}\bra{0}$.
  Next we show $\Gamma = \Gamma' \circ \Lambda$, the composed channel, satisfies the statement.

  Given general input $\rho_1 \otimes \rho_2$, by the guarantee of our disentangler, $\Lambda(\rho_1 \otimes \rho_2)$ satisfies
  \begin{align*}
    \norm{\Lambda(\rho_1 \otimes \rho_2) - \int \ket{\psi}\bra{\psi}^{\otimes k+2k'} d\mu}_1 \le \tilde O\left(\left(\frac{(k+2k')^3}{\ell}\right)^{1/4}\right)\mper
  \end{align*}
  Note that $\Gamma'$ applied to $\int \ket{\psi}\bra{\psi}^{\otimes k+2k'} d\mu$ results in
  \begin{equation}\label{eq:comp_output}
    \int \ket{\psi}\bra{\psi}^{\otimes k} \otimes \left( P_{\text{test}}(\ket{\psi}\bra{\psi})^{k'}  \ket{1}\bra{1} +  (1-P_{\text{test}}(\ket{\psi}\bra{\psi})^{k'}) \ket{0}\bra{0} \right) d\mu \mper
  \end{equation}
  Thus, the composed channel output $\Gamma(\rho_1 \otimes \rho_2))$ is $\tilde O(((k+2k')^3/\ell)^{1/4})$ close, in trace distance, to the state of~(\ref{eq:comp_output}).

  The furthermore part is straightforward. Suppose that $\ket{\psi} = \ket{\phi_1}\otimes \ldots \otimes \ket{\phi_s}$, where $\ket{\phi_i} \in \C^{d_i}$ for $1 \le i \le s$,
  and $\rho_1 = \rho_2 = (\ket{\psi}\bra{\psi})^{\otimes \ell}$. In this case, $\Lambda(\rho_1 \otimes \rho_2) = (\ket{\psi}\bra{\psi})^{\otimes k+2k'}$
  and $\Gamma'(\Lambda(\rho_1 \otimes \rho_2)) = (\ket{\psi}\bra{\psi})^{\otimes k} \otimes \ket{1}\bra{1}$ since $\ket{\psi}$ is a product state
  and product test accepts with probability $1$.
\end{proof}

\paragraph{$\QMA(2)$ Tester --- Gap Amplification for $\QMA(2)$}
The gap amplification of $\QMA(2)$ was first proved in the seminar work of Harrow and Montanaro~\cite{HM13}. Using our template, we provide a conceptually more straightforward proof: Take the old $\QMA(2)$ protocol as the property tester in \cref{fig:dis_and_prop_test}.

\begin{theorem}
Given a language $\calL=(\calL_{\mathrm{yes}}, \calL_{\mathrm{no}})$. Suppose that $\calL \in \QMA(2)$ with completeness $c$ and soundness $s$, where $c-s>1/\poly(n)$. Then, $\calL\in\QMA(2)$ with completeness $c'=1-\exp(-\poly(n))$ and soundness $s'=1/\poly(n)$.
\end{theorem}

\begin{proof}
Let $\calP$ be the protocol for $\calL$ with the promised completeness $c$
and soundness $s$. Therefore, for any fixed input $x$ there is a measurement $M$ acting on a space $\calH^{\otimes 2}$ where $\calH=\C^{d}$, such that,
\begin{align}
        &\exists \sigma\otimes \rho \in \calD(\calH^{\otimes 2}), \; 
            \trace(M (\sigma\otimes\rho) ) \ge c,
                &\text{if } x\in \calL_{\mathrm{yes}}
        \nonumber \\
        &\forall \sigma\otimes\rho \in \calD(\calH^{\otimes 2}), \; 
            \trace(M (\sigma\otimes\rho) ) \le s,
                &\text{if } x\in \calL_{\mathrm{no}.}
        \nonumber
\end{align}
In the new protocol, choose $k=\poly(n)/(c-s)^2$ and $\ell=\poly(k)$ for some large enough polynomial. We ask for two proofs $\ket{\rho_1}, \ket{\rho_2} \in \calD(\calH'^{\otimes\ell}),$ where $\calH'=\C^2 \otimes \calH$. In words, $\calH'$ is $\calH$ with one extra qubit. Apply the disentangler $\Lambda$ from~\cref{thm:channel} on $\rho_1\otimes\rho_2$, obtaining a separable state $\phi=\int d\mu \ket\psi\bra\psi ^{\otimes k}$, such that
\begin{equation}
    \left\|\Lambda(\rho_1\otimes\rho_2)-\int d\mu \ket\psi\bra\psi ^{\otimes k}\right\|_1 = \frac{1}{\poly(n)}.\label{eq:gap-amp-error}
\end{equation}
Consider the new measurement $M'=|01\rangle\langle01|\otimes M$. We apply $M'^{\otimes (k/2)}$ to $\Lambda(\rho_1\otimes\rho_2)$. Accept if more than $(c+s)/2$ fraction of the applications of $M'$ accepts; reject otherwise. Next, we calculate the completeness and soundness of the new protocol.

\emph{Completeness.} Suppose that $x\in\calL_{\mathrm{yes}}$, then the faithful prover will provide
\[
    \ket{\rho_1}=\ket{\rho_2}=\left(\frac{|0,\sigma\rangle+|1,\rho\rangle}{\sqrt{2}}\right)^{\otimes \ell}, \text{ and } \Lambda(\rho_1\otimes\rho_2)=\left(\frac{|0,\sigma\rangle+|1,\rho\rangle}{\sqrt{2}}\right)^{\otimes k}.
\]
Calculating the probability that $M'$ accepts $(|0,\sigma\rangle+|1,\rho\rangle)^{\otimes 2}/2$,
\[
    \trace\left( M' \left(\frac{|0,\sigma\rangle+|1,\rho\rangle}{\sqrt2}\right)^{\otimes 2}
    \right) = \frac{1}{4}\trace( M (\sigma\otimes \rho))\ge c/4.
\]
By Chernoff bound, with probability at least $1-\exp(-\Omega((c-s)^2 k))=1-\exp(-\poly(n))$, the new protocol accepts.

\emph{Soundness.} Suppose that $x\in\calL_{\mathrm{no}}$. 
Calculating the probability that $M'$ accepts $(\alpha|0,\sigma\rangle+\beta|1,\rho\rangle)^{\otimes 2}$ for arbitrary $\alpha,\beta \in \C$ and arbitrary $\sigma, \rho\in \calH$ such that $|\alpha|^2+|\beta|^2=1$,
\[
    \trace( M' (\alpha|0,\sigma\rangle+\beta|1,\rho\rangle)^{\otimes 2}
     =|\alpha\beta|^2\trace( M (\sigma\otimes \rho))\le s/4.
\]
Therefore the probability to accept $\phi$, an arbitrary convex combination of $\ket\psi^{\otimes k}$ is at most $\exp(-\Omega((c-s)^2 k))$ by Chernoff bound. Finally, by (\ref{eq:gap-amp-error}), the probability of accepting $\Lambda(\rho_1\otimes\rho_2)$ is at most $1/\poly(n).$
\end{proof}

%% file: super_swap.tex
\section{The Super Swap and Super Product Tests}\label{sec:super_swap}
In this section, we take another look at the product test as well as the swap test, considering one of the strongest possible generalization of the two.

We start with the more elementary swap test, which is a widely used to test if two quantum states, say $\ket{\psi}$ and $\ket{\phi}$, are equal.
One fundamental limitation of the swap test is that it always accepts with probability at least $1/2$ even if
the states are orthogonal. More precisely, its acceptance probability is $(1+\abs{\braket{\psi}{\phi}}^2)/2$.
Ideally, it would be much more useful to have a test with acceptance probability of $\abs{\braket{\psi}{\phi}}^2$, which is impossible with only one copy for each state.
In the presence of many unentangled copies of $\ket{\phi}$ but just a single copy of $\ket{\psi}$, we show that it
is possible to approach this goal with an arbitrarily small error overcoming the inherent limitation of the swap test. Therefore, we call this test the \emph{super swap} test
and we provide a description of it in~\cref{proc:super_swap}. In particular, this super swap test can be useful
when it is difficult to produce a state $\ket{\psi}$, but much easier to produce copies of $\ket{\phi}$
and we want the tester's acceptance probability to more accurately capture how close $\ket{\psi}$ is to $\ket{\phi}$.
In~\cref{sec:protocol}, the special state $\ket{\psi}$ will be a nonnegative amplitudes state which has a greater cost in the
context of complexity protocols there, whereas $\ket{\phi}$ will have general amplitudes being a cheaper resource in that context.  

\begin{algorithm}{SuperSwap}{($\ket{\psi},\ket{\phi}^{\otimes \ell}$)}
     \begin{enumerate}
           \item Project $\ket{\psi}\ket{\phi}^{\otimes \ell}$ onto the symmetric space $\vee^{\ell+1}(\C^{d})$.
           \item If the projection succeeds \emph{accept};        
            else \emph{reject}.
     \end{enumerate}
     \caption{Super Swap Test}\label{proc:super_swap}
\end{algorithm}

The acceptance probability of the super swap test is established next.

\begin{lemma}\label{lem:super_swap}
  The super swap test accepts with probability 
  \[\frac{\ell \cdot \abs{\braket{\psi}{\phi}}^2}{\ell+1} + \frac{1}{\ell+1}.\]
\end{lemma}

\begin{proof}
  Let $\Pi = (1/(\ell+1)!) \sum_{\pi \in \Sym_{\ell+1}} \pi$ be the projector onto $\vee^{\ell+1}(\C^{d})$. Indeed, we have
  \begin{align*}
    \bra{\psi}\bra{\phi}^{\otimes \ell} \Pi \ket{\psi}\ket{\phi}^{\otimes \ell} = \frac{1}{\ell+1} \braket{\psi}{\psi} \braket{\phi}{\phi}^{\ell} + \frac{\ell}{\ell+1} \abs{\braket{\psi}{\phi}}^2 \braket{\phi}{\phi}^{\ell-2}\mcom
  \end{align*}
  concluding the proof.
\end{proof}

At first glance, it may seem inconvenient to assume multiple ($\ell$-many) unentangled copies of $\ket{\phi}$.
However, due to our disentangler channel, we can enforce a distribution over product states $\ket{\phi}^{\otimes \ell}$ by assuming
only bipartite unentanglement.

Next we turn to the product test which checks whether a state is close to a $k$-partite product state~\cite{HM13}.
It has a similar drawback to the usual swap test, namely, it always accepts with probability at least $1/2$
even if the state $\ket{\psi}$ is very far from product. As before, we will arbitrarily improve the soundness of the product test
by having multiple unentangled copies. We call this new test the \emph{super product} test and we describe it in~\cref{proc:super_prod}.

\begin{algorithm}{SuperProduct}{($\ket{\psi},(\ket{\phi_1}\ldots\ket{\phi_k})^{\otimes \ell}$)}
     \begin{enumerate}
           \item Project $\ket{\psi}(\ket{\phi_1}\ldots\ket{\phi_k})^{\otimes \ell}$ onto the symmetric space $\vee^{\ell+1}((\C^{d})^{\otimes k})$.
          \item If the projection succeeds \emph{accept};        
            else \emph{reject}.
     \end{enumerate}
     \caption{Super Product Test}\label{proc:super_prod}
\end{algorithm}

\begin{lemma}
  The super product test accepts with probability
    \[ 
  \frac{\ell}{(\ell+1)} \cdot \abs{\braket{\psi}{\phi_1}\ldots\ket{\phi_k}}^2 + \frac{1}{(\ell+1)}\mper
    \] 
\end{lemma}

\begin{proof}
  We view each copy of the state $\ket{\phi_1}\ldots\ket{\phi_k}$ as a single state $\ket{\phi}$
  and apply the super swap test to $\ket{\psi}$ and $\ket{\phi}^{\otimes \ell}$. The acceptance  
  probability of the super product test now follows from~\cref{lem:super_swap}.
\end{proof}

Analogously, it may seem inconvenient to assume multiple ($\ell$-many) unentangled copies of $\ket{\phi_1} \ldots \ket{\phi_k}$.
However, that is not an issue by \cref{cor:separable-states}: We can enforce a distribution over product states $(\ket{\phi_1} \ldots \ket{\phi_k})^{\otimes \ell}$ by assuming only $2$ unentangled states.

%% file: protocol.tex
\section{Gap Amplification for $\QMA^+(k)$ up to Criticality and Almost-$\QMA(k)=\NEXP$}\label{sec:protocol}

In the previous section, we described a very strong version of swap test and product test, noting that our disentangler channel has a good synergy with the new tests to overcome the drawbacks in their original versions.
In this section, we put the tools in the context of quantum Merlin-Arthur games with unentangled provers,
establishing our main complexity results~\cref{thm:almost_gen_amp,thm:qma3}.

\subsection{Gap Amplification for $\QMA^+(k)$ up to Criticality}

The gap amplification for $\QMA^+(k)$ is much less straightforward than $\QMA(2)$.
Indeed, a full gap amplification would imply $\QMA(2) = \NEXP.$
To give our half gap amplification promised in~\cref{thm:qma3}, we start by showing how to
simulate a $\QMA^+(k)$ protocol $\calP$ given the following kinds of proofs:
\begin{enumerate}
    \item\label{enu:nonnegative-state} one nonnegative-amplitudes proof $\ket{\psi}$;
    \item\label{enu:many-general-states} abundant equal copies of an arbitrary proofs over reals $\ket{\phi}$.
\end{enumerate}   
Note we are relaxing $k$ nonnegative-amplitudes proofs in a $\QMA^+(k)$ protocol with only one nonnegative-amplitudes
proof and general-amplitudes states.
The motivation is, roughly, to remove as many nonnegative-amplitudes proofs in a $\QMA^+(k)$ protocol as possible, so we get
closer to a general $\QMA(k)$ protocol. 

We will check whether  $\ket{\phi}^{\otimes k}$
is close to $\ket{\psi}$. Either they are close and then we can use the many copies of $\ket{\phi}^{\otimes k}$
to simulate $\calP$, or else they are far apart and an application of the super product test can detect this
condition. A description of this simulation procedure is given in~\cref{proc:sym_protocol}, which we denote
as the symmetric simulator (since it assumes many equal copies of $\ket{\phi}$).

\begin{algorithm}[H]{SymSimulator}{($\calP$, $\ket{\psi} = \sum_i \beta_i \ket{i} \colon \beta_i \ge 0, \ket{\phi}^{\otimes 2 k \ell}$)}
     \begin{itemize}
       \item If $\textup{SuperProduct}(\ket{\psi},(\ket{\phi}^{\otimes k})^{\otimes \ell})$ fails, then \emph{reject}.
       \item For $i=1,\ldots,\ell$ \begin{enumerate}
          \item Run the $\QMA^+(k)$ protocol $\calP$ on a new copy of $\ket{\phi}^{\otimes k}$.
          \item If protocol rejects, then \emph{reject}.
       \end{enumerate}
       \item \emph{Accept}.
     \end{itemize}
     \caption{Symmetric Simulator}\label{proc:sym_protocol}
\end{algorithm}

We now analyze the completeness and soundness of this simulation.

\begin{lemma}\label{lem:sym_simulator}
  Suppose $\calP$ is a $\QMA^+(k)$ protocol with completeness $c$ and soundness $s$.
  Let $p(n)$ be a non-decreasing function such that $p(n) \ge C_0$ for a sufficiently large
  constant $C_0 > 0$. If $\ell \ge 8 p(n)^2 \ln(2)$ and $s \le 1/8p(n)^2$, then 
  SymSimulator has completeness $c^\ell$ and soundness at most $1/2 + 1/p(n)$.
\end{lemma}

\begin{proof}
  In the completeness case, we can assume that the proofs $\ket{\phi}$ have nonnegative amplitudes
  and $\ket{\psi} = \ket{\phi}^{\otimes k}$. Thus, SymSimulator accepts with probability at least $c^\ell$.

  Now, suppose that we are in the soundness case. Set $\epsilon =\abs{\braket{\psi}{\phi^{\otimes k}}}^2$.
  By~\cref{lem:super_product}, the super product test accepts with probability
    \[ 
  \left( \frac{ \epsilon\ell}{\ell+1} + \frac{1}{\ell+1} \right)\mper
    \] 
      Since $\ell \ge 2p(n)$, if $\epsilon < 1/2 + 1/2p(n)$, then the acceptance probability due to the super product test alone is at most $1/2+1/p(n)$
  and we are done. Therefore, from now on, we assume that $\epsilon \ge 1/2 + 1/2p(n)$. 
  
  Suppose $\ket{\phi} = \sum_i \alpha_i \ket{i}$, and let $\ket{\phi_+} = \sum_{i} \abs{\alpha_i} \ket{i}$. Thus, $\ket{\phi_+}$ is a valid nonnegative-amplitudes state.
  Since $\ket{\psi}$ has nonnegative amplitudes by assumption, we should have 
    \begin{equation}\label{eq:phi-plus-corr}
        |\langle \psi \mid \phi_+^{\otimes k} \rangle|^2 \ge |\langle \psi \mid \phi^{\otimes k} \rangle|^2 = \epsilon.
    \end{equation}
  This is because the latter inner product incurs some cancellations due to negative values, which are avoided in the former inner product. 
  (\ref{eq:phi-plus-corr}) together with \cref{fact:overlap-triangle-ineq} implies that  
    \begin{equation*}
        |\langle \phi^{\otimes k}, \phi_+^{\otimes k}\rangle|^2
            \ge 2\epsilon-1.
    \end{equation*}
  Since we are assuming $\epsilon > 1/2$, the trace distance between $\ket\phi ^{\otimes k}$ and $\ket{\phi_+}^{\otimes k}$ can be bounded as below
    \begin{equation}
        D(\phi^{\otimes k}, \phi_+^{\otimes k})\le 2\sqrt{\epsilon(1-\epsilon)}
    \end{equation}
  Note that $\calP$ accepts $|\phi_+^{\otimes k}\rangle$ with probability at most $s$ by the soundness of $\calP$.
  Therefore, each execution of the protocol $\calP$ on $\ket\phi ^{\otimes k}$ accepts with probability, by \cref{fact:trace_norm_acc}, at most 
  \[\min\set{1,2\sqrt{\epsilon(1-\epsilon)} + s}.\]
  The overall soundness of SymSimulator becomes
    \[ 
  \left(\epsilon \frac{\ell}{\ell+1} + \frac{1}{\ell+1} \right) \left( \min\set{1,2\sqrt{\epsilon(1-\epsilon)} + s} \right)^{\ell} .
    \] 
 Now take  $\epsilon \ge 1/2 + 1/2p(n)$, and compute, we have
    \[ 
  2\sqrt{\epsilon(1-\epsilon)} \le 2 \sqrt{\frac{1}{4} - \frac{1}{4 p(n)^2}} \le 1 - \frac{1}{2 p(n)^2} + O\left(\frac{1}{p(n)^4}\right) \le 1 - \frac{1}{4p(n)^2}\mcom
    \] 
  where the last inequality relies on $p(n) \ge C_0$ for a large enough constant $C_0 > 0$.
  Using that $s \le 1/8p(n)^2$ and $\ell  \ge 8 p(n)^2 \ln(2)$, the final acceptance probability is
    \[ 
  \left(2\sqrt{\epsilon(1-\epsilon)} + s \right)^{\ell} \le \left(1 - \frac{1}{8p(n)^2} \right)^\ell \le \frac{1}{2},
    \] 
  concluding the proof.
\end{proof}

To remove  \ref{enu:many-general-states} the symmetric assumption of having multiple identical copies of $\ket{\phi}$ in SymSimulator, we use the PAPO channel $\Lambda$ and the PAPO channel takes
just two unentangled proofs $\ket{\phi'}$ and $\ket{\phi''}$ (of arbitrary amplitudes) as its input. In other words, we now simulate a $\QMA^+(k)$ protocol $\calP$ with:
\begin{description}
    \item \; (i)\, one nonnegative-amplitudes proof $\ket{\psi}$;
    \item (ii')\,  two general states $\ket{\phi'}, \ket{\phi''}$.
\end{description} 
A formal description of the new simulation is given in~\cref{proc:protocol}.

\begin{algorithm}[H]{Simulator}{($\calP$, $\ket{\psi} = \sum_i \beta_i \ket{i} \colon \beta_i \ge 0, \ket{\phi'}, \ket{\phi''}$)}
     \begin{itemize}
       \item Let $\rho$ be the output of our disentangler $\Lambda(\phi'\otimes \phi'')$ (i.e. \cref{thm:channel}).
       \item If $\textup{SymSimulator}(\calP, \ket{\psi},\rho)$ accepts, then \emph{accept}; else \emph{reject}.
     \end{itemize}
     \caption{Protocol Simulator}\label{proc:protocol}
\end{algorithm}

The analysis of~\cref{proc:protocol} is similar to that of \cref{lem:sym_simulator}. Therefore, instead of presenting an analysis of~\cref{proc:protocol} in isolation, we now apply this simulation for a $\QMA^+(k)$ protocol $\calP$ that solves a $\NEXP$-complete problem.  
In particular, we will need the following characterization of $\QMA^+(2)$ from~\cite{JW23} as shown in the following theorem.

\begin{theorem}[\cite{JW23}]\label{thm:qma2_plus_nexp}
  $\QMA^+(2) = \NEXP$.
\end{theorem}

\cref{proc:protocol} gives rise to a protocol for NEXP that improves the above theorem in two aspects. First, the new protocol uses three unentangled proofs among which only one is required to have nonnegative amplitudes. Second, the completeness and soundness gap of this protocol is about $1/2$. This seemingly mediocre gap is in fact a critical point, which we discuss in the next section.

\TheoQMAThree*

\begin{proof}
  From~\cref{thm:qma2_plus_nexp}, we apply the standard gap amplification by asking for more unentangled proofs to obtain a $\QMA^+(k)$ protocol $\calP$ with completeness $c = 1-\exp(-\poly(n))$
  and soundness $s=\exp(-\poly(n))$, where $k=\poly(n)$. Simulate $\calP$ using~\cref{proc:protocol}.
  By \cref{thm:channel}, $\rho = \Lambda(\phi'\otimes \phi'')$ is $1/\poly(n)$-close to a convex combination of
  product states $\int \ket{\phi} \bra{\phi}^{\otimes 2k\ell} d\mu$ with $\ell = \poly(n)$. Invoking the symmetric simulator, by~\cref{lem:sym_simulator},
  the completeness becomes $c^\ell \ge 1-\exp(-\poly(n))$ and the soundness $1/2+1/\poly(n)$ for a suitable choice of polynomial $\ell=\poly(n)$.
\end{proof}

\subsection{Almost-$\QMA^\R(k) = \NEXP$}
Next, we show how to go from the nonnegative amplitudes assumptions to almost general amplitudes.
Recall that the $\ell_2$-sign bias of a state $\ket{\psi} = \sqrt{a} \ket{\psi_+} + \sqrt{1-a}\ket{\psi_{-}}$,
where $\ket{\psi_+}$ and $\ket{\psi_{-}}$ are the normalized nonnegative and negative amplitudes parts of $\ket{\psi}$,
is defined as $\abs{a-(1-a)}$ (see~\cref{def:ell2_sign_bias}). 

\TheoAlmostGenAmp*

\begin{proof}
  We start with the $\QMA^+(3)$ protocol from~\cref{thm:qma3} with two general proofs $\ket{\phi'},\ket{\phi''}$
  and only one nonnegative proof $\ket{\psi}$. Let $M$ be the verifier measurement.
  In the completeness case, we can assume that $\ket{\psi}$ has nonnegative amplitudes so we proceed
  to analyze the soundness case.

  In the \textup{almost}-$\QMA^\R(3)$ protocol, $\ket{\psi}$ will no-longer be assumed to have nonnegative
  amplitudes. Instead, we write  $\ket{\psi} = \sqrt{a} \ket{\psi_+} + \sqrt{1-a} \ket{\psi_{-}}$,
  where $\ket{\psi_+}$ and $\ket{\psi_{-}}$ are its nonnegative- and negative-amplitudes normalized states.
  Without loss of generality, suppose that $a \ge 1/2$.
  Furthermore, under the $\ell_2$-sign bias assumption, we may assume that 
  \begin{equation}\label{eq:a-bias}
      a \ge 1/2 + \sqrt{100/p(n)}.
  \end{equation}
  Let $\ket{\phi'}$ and $\ket{\phi''}$ be some quantum states (ignoring the $\ell_2$-bias requirement) as to be used in the simulation~\cref{proc:protocol}.
  The combined proofs of the \textup{almost}-$\QMA^\R(3)$ protocol can be expressed as $\ket{\xi} = \sqrt{a} \ket{\xi_0} + \sqrt{1-a} \ket{\xi_1}$, where
  $\ket{\xi_0} = \ket{\phi'} \otimes \ket{\phi''} \otimes \ket{\psi_+}$ and $\ket{\xi_1} = \ket{\phi'} \otimes \ket{\phi''} \otimes \ket{\psi_{-}}$. 
  Denote s the soundness of $\QMA^+(3)$ protocol from~\cref{thm:qma3}. Then we can assume 
  \begin{equation}\label{eq:soundness-qma3}
  s \le 1/2 + 6/p(n).
  \end{equation}
  Calculating the accepting probability of $M$ on $\xi$,
  \begin{align}
    \bra{\xi} M \ket{\xi} 
        &= a \bra{\xi_0} M \ket{\xi_0} + (1-a) \bra{\xi_1} M \ket{\xi_1} \nonumber \\
        &\qquad\qquad +\sqrt{a(1-a)} \bra{\xi_0} M \ket{\xi_1} + 
        \sqrt{a(1-a)} \bra{\xi_1} M \ket{\xi_0}
        \nonumber \\
        &\le s + \sqrt{a(1-a)} \left(\bra{\xi_0} M \ket{\xi_0} + \bra{\xi_1} M \ket{\xi_1}\right) 
        \nonumber\\
        &\le (1+ 2\sqrt{a(1-a)})s. \label{eq:soundness-almost-qma}
  \end{align}
  where the first inequality follows from $M$ being PSD, \ie since $(\bra{\xi_0}-\bra{\xi_1}) M (\ket{\xi_0}-\ket{\xi_1}) \ge 0$
  implies $\bra{\xi_0} M \ket{\xi_0} + \bra{\xi_1} M \ket{\xi_1} \ge \bra{\xi_0} M \ket{\xi_1} + \bra{\xi_1} M \ket{\xi_0}$.
  By~(\ref{eq:a-bias}) and~(\ref{eq:soundness-qma3}), we have
    \[ 
  (1+ 2\sqrt{a(1-a)})s \le \left(2- \frac{8}{p(n)}\right) s \le\left(2- \frac{8}{p(n)}\right) \left(\frac{1}{2} + \frac{1}{p(n)}\right) \le 1 - \frac{2}{p(n)}.
    \]

Note that by a suitable choice of polynomial $p(n)$ and  the initial completeness $c=1-\exp(-\poly(n))$ of the $\QMA^+(3)$ protocol of~\cref{thm:qma3},
we obtain a gap of $\Omega(1/p(n))$. To conclude the proof, we apply standard gap amplification using $k=\poly(p(n))$ proofs in $\textup{almost-QMA}^\R(k)$.
\end{proof}

We emphasize an important observation following from the above analysis: The ``half'' gap amplification in~\cref{thm:qma3} is almost optimal. A larger gap in~\cref{thm:qma3} by an additive term $1/\poly(n)$ (\eg if the soundness was at most $1/2-1/\poly(n)$) would allow us to completely discard the $\ell_2$-sign bias assumption in \cref{thm:almost_gen_amp}, showing $\NEXP=\QMA^{\mathbb{R}}(k)$. This can be easily seen in (\ref{eq:soundness-almost-qma}), when $s<1/2-1/\poly(n)$, the RHS will be at most $1-1/\poly(n).$ It means that $s=1/2\pm 1/\poly(n)$ in~\cref{thm:qma3} is a critical point. In the case that $\QMA(k)^\R\not=\NEXP$, there is a sharp phase transition.